\documentclass[a4paper,12pt]{article}

\title{New distance bounds for one-generator quasi-cyclic codes with applications to Hermitian LCD codes and entanglement-assisted quantum codes}
\author{
Kanat Abdukhalikov \\	
Department of Mathematical Sciences, \\
UAE University, PO Box 15551, Al Ain, UAE\\
Email: abdukhalik@uaeu.ac.ae \bigskip  \\  
Duy Ho \\
Department of Mathematical Sciences, \\
UAE University, PO Box 15551, Al Ain, UAE\\ 
Email: duyho@uaeu.ac.ae \bigskip  \\  
Rasha M. Shat \\
Department of Mathematical Sciences, \\
UAE University, PO Box 15551, Al Ain, UAE\\ 
Email: 201450047@uaeu.ac.ae \bigskip  \\  
}
\date{ }

\usepackage{amsthm,amsmath,amssymb} 

\DeclareMathOperator{\wt}{wt}
\DeclareMathOperator{\bsupp}{bsupp} 
\usepackage{url} 
\usepackage{cases} 
\usepackage{makecell}
\usepackage{enumerate}
\usepackage{tikz-cd}
\usepackage{float}
\usepackage{booktabs}
\usepackage{float}
\usepackage{longtable}
\usepackage{pdflscape}

\usepackage[table]{xcolor}
\usepackage{tikz-cd}

\usepackage{totcount}

\newtotcounter{grasslcode}
\newcommand{\nextgrasslno}{%
  \refstepcounter{grasslcode}%
  \arabic{grasslcode}%
}

\newtotcounter{appendixcode}
\newcommand{\nextcodeno}{%
  \refstepcounter{appendixcode}%
  \arabic{appendixcode}%
}

\begin{document} 

\maketitle

\theoremstyle{plain} 
\newtheorem{lemma}{Lemma}[section] 
\newtheorem{theorem}[lemma]{Theorem}
\newtheorem{corollary}[lemma]{Corollary}
\newtheorem{proposition}[lemma]{Proposition}

\theoremstyle{definition}
\newtheorem{definition}[lemma]{Definition}
\newtheorem{remark}[lemma]{Remark}
\newtheorem{example}[lemma]{Example}

\newcommand{\eps}{\varepsilon}
\newcommand{\inprod}[1]{\left\langle #1 \right\rangle}
\newcommand{\la}{\lambda} 
\newcommand{\al}{\alpha}
\newcommand{\om}{\omega} 
\newcommand{\gam}{\gamma}
\newcommand{\be}{\beta}
\newcommand{\sig}{\sigma} 
\newcommand{\concat}{\mathbin{\square}}

\begin{abstract}
We introduce lower bounds on the minimum distance of one-generator quasi-cyclic codes of arbitrary index, obtained by puncturing and shortening the code at the level of its blocks. 
For balanced codes, these bounds form a nondecreasing chain
whose first term is the Jensen bound. 
We apply these bounds to conduct a computer search and obtain many Hermitian LCD codes with good parameters.  
Among them there are six codes of length at most $30$ which improve the corresponding lower bounds of Araya and Harada. 
These codes further give maximal-entanglement entanglement-assisted quantum codes, which we compare with the parameters currently recorded in the literature.
\end{abstract}

\textbf{Keywords}:  Linear codes, quasi-cyclic codes,    Hermitian LCD codes,  entanglement-assisted quantum codes.

\textbf{Mathematics subject classification}: 94B05, 94B15, 94B60.

\section{Introduction}

Quasi-cyclic codes form an interesting class of linear codes that generalize cyclic codes. 
A linear code of length $n=m\ell$ is quasi-cyclic of index $\ell$ if it is invariant under the shift $T^{\ell}$ by $\ell$ coordinates, the case $\ell=1$ recovering cyclic codes. 
The class of quasi-cyclic codes was already studied in the late 1960s
\cite{townsend1967}, and was shown early on to be asymptotically good
\cite{chen1969}. 

Quasi-cyclic codes are appealing in practice. A quasi-cyclic code over $F=\mathbb{F}_q$ of index $\ell$ and length $m\ell$ can be described by a small number of polynomials in $F[x]/\langle x^m-1\rangle$, so that large families can be generated and searched exhaustively, while enough algebraic structure is retained to give nontrivial lower bounds on the minimum distance. 
For this reason, quasi-cyclic codes remain a standard source of record-holding linear and quantum codes; see
the online table \cite{grassl}.

There are different approaches one can choose to study quasi-cyclic codes. 
One example is the  graph-theoretic approach in the quasi-cyclic LDPC literature, where a code is analyzed through the associated Tanner graph \cite{fossorier2004}. 
In our paper however, we follow the algebraic approach initiated by Ling and Sol\'e \cite{ling2001}.
We view a quasi-cyclic code of index $\ell$  as a module over
$R=F[x]/\langle x^m-1\rangle$. 
When $\gcd(q,m)=1$, a quasi-cyclic code can be decomposed by the Chinese Remainder Theorem (CRT) into a direct sum of \emph{constituent} codes of length $\ell$ over extensions of $F$. 
Complementary to this is the concatenated description of Jensen \cite{jensen1985}, in which the constituents
appear as outer codes and minimal cyclic codes as inner codes. For a survey, we refer to \cite{guneri2021}.

Within the algebraic approach, the distance bounds available for quasi-cyclic codes mostly come from the CRT decomposition or the concatenated description of the code. 
The Jensen bound \cite{jensen1985} combines the
distances of the constituents with those of nested sums of minimal cyclic codes. 
The spectral method of Semenov and Trifonov \cite{semenov2012} uses the
eigenvalues of the generator polynomial matrix, and was substantially sharpened by Luo, Ezerman, Ling and \"Ozkaya \cite{luo2024}.
Further bounds along these lines were obtained by G\"uneri and \"Ozbudak \cite{guneri2012,guneri2013}, and by \"Ozbudak and \"Ozkaya \cite{ozbudak2024}. A systematic comparison
of the known bounds is given in \cite{ezerman2021}.

An alternative way to derive distance bounds for quasi-cyclic codes is to  inspect the generators of the code, rather than  its constituents. 
Perhaps the earliest bound of this kind is due to Lally \cite{lally2003}, which is a product of the distances of two auxiliary codes obtained from the generating polynomials.
Closer to our purpose is a recent argument of Guan, Li, Liu and Ma in
\cite[Theorem~1]{guan2023b}. 
In their investigation of the symplectic distance of
one-generator quasi-cyclic codes of even index, the authors of \cite{guan2023b} consider a condition on the generator polynomials and obtain a bound  in terms of the minimum distance of a single cyclic code. We call a code satisfying their condition \emph{balanced}.

Generalizing their argument, in this paper we develop distance bounds valid for every one-generator quasi-cyclic code. 
Our bounds come from puncturing and shortening
quasi-cyclic codes at the generator block level. 
Specializing back to balanced codes, we obtain a nondecreasing chain of block-support bounds. 
We further show that the first term of this chain is the Jensen bound, so our bounds are generally better than the Jensen bound for balanced one-generator quasi-cyclic codes.

We then demonstrate the applicability of these bounds in the construction of Hermitian linear complementary dual (LCD) codes. 
LCD codes were introduced by Massey \cite{massey1992} and have since attracted attention through their use in
countermeasures against side-channel and fault injection attacks
\cite{carlet2016}. 
In \cite{carlet2018}, it was shown that every linear code over $\mathbb F_{q}$ with $q>4$ is
equivalent to a Hermitian LCD code.
This motivates the ongoing work to search for quaternary Hermitian LCD codes with good parameters in 
\cite{araya2020,araya2022,araya2024,ishizuka2020}. 
In Section~4, we give a polynomial description of the Hermitian
hull of a one-generator quasi-cyclic code and characterize the Hermitian LCD property.
These results strengthen the corresponding results of  \cite{guneri2023, guan2023}. 
Based on the Hermitian LCD criterion and the obtained bounds, we conduct a computer search and obtain many Hermitian LCD codes with good parameters. 
We highlight six quaternary Hermitian LCD codes of length at most $30$ which improve the corresponding lower bounds in the Araya-Harada table \cite[Table~6]{araya2024}.

Finally, Hermitian LCD codes are of interest in quantum coding theory. An $[n,k,d]$ Hermitian LCD code over $\mathbb F_{q^2}$ yields a maximal-entanglement entanglement-assisted quantum error-correcting code (EAQECC) with parameters $[[n,k,d;n-k]]_q$, see \cite[Corollary 3.4]{guenda2018} and \cite{wildebrun2008}.
We record the EAQECCs arising in this way from our Hermitian LCD codes and compare them with the entries of Grassl's table \cite{grassl} and with the recent constructions of Li, Ezerman, Li, Ling and Sun \cite{LELLS26}.

The paper is organized as follows. 
Section~2 provides the necessary background
on quasi-cyclic codes, their Chinese Remainder Theorem decomposition, the concatenated description and the Jensen bound. 
Section~3 introduces the block-punctured and block-shortened codes, establishes the block-support bound, and analyzes the balanced case. Section~4 develops the Hermitian hull and the Hermitian LCD criterion for one-generator quasi-cyclic codes and presents the codes obtained from our search. 
Section~5 is devoted to the resulting
maximal-entanglement entanglement-assisted quantum codes. 
Generator polynomials of all codes obtained in our search are listed in the appendix.

All computations in this paper have been carried out with the computer algebra system \textsc{Magma} \cite{magma}.

\section{Preliminaries} 

\subsection{Background on linear and quasi-cyclic codes}

Let $F=\mathbb{F}_q$ be the finite field with $q$ elements, where $q$ is
a prime power. 
A \textit{linear code} $C$ of length $n$ over $F$ is a
subspace of the vector space $F^n$, and its elements are called
\textit{codewords}. 
The \textit{weight} $\wt(\mathbf{c})$ of
$\mathbf{c}\in F^n$ is the number of nonzero coordinates of $\mathbf{c}$.
If $C\neq\{\mathbf{0}\}$, the \textit{minimum distance} of $C$ is
\[
d(C):=\min\{\wt(\mathbf{c}) : \mathbf{0}\neq\mathbf{c}\in C\}.
\]
For the zero code we will use the convention $d(\{\mathbf{0}\})=+\infty$.

 Let $T$ be the standard cyclic shift operator on $F^n$. 
 A linear code is said to be \textit{quasi-cyclic of index $\ell$ (QC)} if it is invariant under $T^\ell$. 
 We can assume that $\ell$ divides $n$. If $\ell= 1$, then the QC code is a cyclic code.

Let $R = F[x]/\langle x^m-1\rangle$. We recall that cyclic codes of length $m$ over $F$ can be considered as ideals of $R$.

Let $n = m\ell$ and let $C$ be a linear quasi-cyclic code of length $m \ell$ and index $\ell$ over $F$. Let
 \[
\mathbf{c} =(c_{0,1},  c_{0,2}, \dots,  c_{0,\ell},\, c_{1,1},  c_{1,2}, \dots,  c_{1,\ell}, \dots, c_{m-1,1},  c_{m-1,2}, \dots,  c_{m-1,\ell})
 \]
denote a codeword in $C$. Define a map $\varphi: F^{m\ell} \rightarrow R^\ell$ by
\[
\varphi(\mathbf{c}) = (c_1(x), c_2(x), \dots, c_{\ell}(x)) \in R^\ell,
\]
where
\[
c_j(x) = c_{0,j}+c_{1,j}x+c_{2,j}x^2+ \dots + c_{m-1,j}x^{m-1} \in R.
\]
The following lemma is well-known.
 \begin{lemma}[\cite{lally2001,ling2001}] The map  $\varphi$ induces a one-to-one correspondence between quasi-cyclic
 codes over $F$ of index $\ell$ and length $m \ell$ and linear codes over $R$ of length $\ell$.
\end{lemma}

We identify $C$ with $\varphi(C)\subseteq R^\ell$. For
$a(x)\in R$, we write $\wt(a)$ for the number of nonzero coefficients of
its unique representative of degree less than $m$. Under the
identification above, we have
\begin{equation}\label{blockweight}
\wt(\mathbf{c})=\sum_{j=1}^{\ell}\wt\big(c_j(x)\big),
\end{equation}
so weights and minimum distances of quasi-cyclic codes can be computed in $R^\ell$.

For $\mathbf g_1,\dots,\mathbf g_r\in R^\ell$, we denote by
$\langle \mathbf g_1,\dots,\mathbf g_r\rangle$ the $R$-submodule of
$R^\ell$ generated by $\mathbf g_1,\dots,\mathbf g_r$, which is also called \emph{$r$-generator} quasi-cyclic code. 
One-generator quasi-cyclic codes were studied in \cite{abd2023,seguin2004}. 
The structure of quasi-cyclic codes of index $2$ and $3$ was investigated in \cite{abd2026,abd2025}.
Linear complementary dual and complementary pairs of quasi-cyclic codes of index $2$ were investigated in
\cite{abd2026-cc,abd2026-amc}.

\subsection{The Chinese Remainder Theorem decomposition of quasi-cyclic codes} 

Throughout the paper we will assume that $\gcd(q,m)=1$. Let
\[
  x^m-1=f_1(x)f_2(x)\cdots f_\nu(x)
\]
be the factorization of $x^m-1$ into irreducible polynomials in $F[x]$.
Since $\gcd(q,m)=1$, the factors $f_1,\ldots,f_\nu$ are distinct. 
For $1\le i\le \nu$, let
\[
  E_i:=F[x]/\langle f_i(x)\rangle.
\]
Let $\xi$ be a primitive $m$-th root of unity in an extension field of $F$. 
For each $i$, let $u_i$ be the smallest nonnegative integer such that
$f_i(\xi^{u_i})=0$. Then $E_i\cong F(\xi^{u_i})$. By the Chinese
Remainder Theorem, we have the ring isomorphism
\begin{equation}\label{crtR}
  R=F[x]/\langle x^m-1\rangle
  \cong E_1\oplus\cdots\oplus E_\nu,
  \qquad
  a(x)\longmapsto \big(a(\xi^{u_1}),\ldots,a(\xi^{u_\nu})\big).
\end{equation}
This isomorphism extends naturally to $R^\ell$, which implies that
\begin{equation}\label{crtRl}
  R^\ell\cong E_1^\ell\oplus\cdots\oplus E_\nu^\ell .
\end{equation}
Hence, a quasi-cyclic code $C\subseteq R^\ell$ decomposes as
\begin{equation}\label{crtC}
  C\cong C_1\oplus\cdots\oplus C_\nu,
\end{equation}
where $C_i$ is a linear code of length $\ell$ over $E_i$, for each
$i$. These length $\ell$ linear codes over the various extension fields
of $F$ are called the \emph{constituents} of $C$. 
Let $C \subseteq R^\ell$ be $r$-generated as an $R$-module by
\[
\left\{
\bigl(a_{1,1}(x),\ldots,a_{1,\ell}(x)\bigr),\ldots,
\bigl(a_{r,1}(x),\ldots,a_{r,\ell}(x)\bigr)
\right\}
\subseteq R^\ell .
\]
Then, for \(1 \le i \le \nu\), we have
\begin{equation}\label{spanconstituent}
  C_i=\operatorname{span}_{E_i}
  \left\{
  \bigl(a_{b,1}(\xi^{u_i}),\ldots,a_{b,\ell}(\xi^{u_i})\bigr):
  1\le b\le r
  \right\}
  \subseteq E_i^\ell .
\end{equation}

\subsection{Concatenated decomposition and the Jensen bound}

We keep the notation of the previous subsection. 
Each extension field
$E_i$ is isomorphic to a minimal cyclic code of length $m$ over $F$,
namely the cyclic code whose check polynomial is $f_i(x)$. 
Let
$\theta_i$ be the generating primitive idempotent of this minimal
cyclic code. Then the isomorphism is given by
\[
  \varphi_i:\langle\theta_i\rangle\longrightarrow E_i,
  \qquad
  a(x)\longmapsto a(\xi^{u_i}),
\]
with inverse map
\[
  \psi_i:E_i\longrightarrow \langle\theta_i\rangle,
  \qquad
  \delta\longmapsto \sum_{\tau=0}^{m-1} a_\tau x^\tau,
\]
where
\[
  a_\tau=
  \frac{1}{m}\operatorname{Tr}_{E_i/F}
  \bigl(\delta\xi^{-\tau u_i}\bigr).
\] 
If $C_i \subseteq E_i^\ell$ is a linear code of length $\ell$ over
$E_i$, we denote by $  \langle\theta_i\rangle\concat C_i $
the concatenation of $C_i$ with $\langle\theta_i\rangle$. 
This concatenation is carried out by the map $\psi_i$ extended to
$E_i^\ell$.
In other words, $\psi_i$ is applied to each symbol of the
codeword in $C_i$.

Jensen gave the following concatenated description of quasi-cyclic codes.

\begin{theorem} \label{concatdescription}
\begin{enumerate}
\item[(a)]
Let $C$ be an $R$-submodule of $R^\ell$. Then, for some subset
$I\subseteq\{1,\ldots,\nu\}$, there exist linear codes
$C_i\subseteq E_i^\ell$ such that
\[
  C = \bigoplus_{i\in I} \langle\theta_i\rangle\concat C_i .
\]

\item[(b)]
Conversely, let $C_i\subseteq E_i^\ell$ be a linear code for each
$i\in I\subseteq\{1,\ldots,\nu\}$. Then
\[
  C = \bigoplus_{i\in I} \langle\theta_i\rangle\concat C_i
\]
is an $F$-linear quasi-cyclic code of length $m\ell$ and index
$\ell$.
\end{enumerate}
\end{theorem}

It was proved in \cite[Theorem~4.1]{guneri2013} that, for a given
quasi-cyclic code $C$, the constituents $C_i$ in \eqref{crtC} and the
outer codes $C_i$ in the concatenated structure of
Theorem~\ref{concatdescription} are the same.

The concatenated structure in Theorem~\ref{concatdescription} yields a minimum distance bound for quasi-cyclic codes \cite[Theorem 4]{jensen1985}. 

\begin{theorem}[Jensen bound]\label{jensen}
Let $C$ be an $F$-linear quasi-cyclic code of length $m\ell$ and
index $\ell$. Let
\[
  I(C):=\{i:C_i\ne\{\mathbf 0\}\} =\{i_1,\ldots,i_t\},
\]
where the indices are ordered so that
\[
  0<d(C_{i_1})\le d(C_{i_2})\le\cdots\le d(C_{i_t}).
\]
Then the concatenated decomposition of $C$ can be written as
\[
  C = \bigoplus_{z=1}^{t} \langle\theta_{i_z}\rangle\concat C_{i_z},
\]
and
\[
  d(C) \ge d_J(C) :=
  \min_{1\le z\le t} 
  \left\{  d(C_{i_z})\, 
    d\bigl(
      \langle\theta_{i_1}\rangle
      \oplus\cdots\oplus
      \langle\theta_{i_z}\rangle
    \bigr)
  \right\}.
\]
\end{theorem}

\section{Block-support distance bounds for one-generator quasi-cyclic codes}

We   consider a one-generator
quasi-cyclic code of index $\ell$,
\[
C=\langle\mathbf g\rangle\subseteq R^\ell,
\qquad
\mathbf g:=(g_1(x),\dots,g_\ell(x)).
\]
Throughout, we identify an element of $R$ with its unique representative in $F[x]$ of degree less than $m$.

For $1\le j\le \ell$, let 
\[
\gamma_j(x):=\gcd\!\big(g_j(x),x^m-1\big),\qquad
\eta_j(x):=\frac{x^m-1}{\gamma_j(x)} .
\]
Let 
\[
g:=\gcd(g_1,\dots,g_\ell,x^m-1).
\]
Then $\dim C=m-\deg g$, see \cite{seguin2004}. 

\begin{lemma}\label{normalization}
There exists a unit $u\in R$ with $ug_1\equiv \gamma_1 \pmod{x^m-1}$. In particular, given $C=\langle\mathbf g\rangle$, we may assume without
loss of generality that $g_1 \mid x^m-1$, and \[g=\gcd(g_1,\dots,g_\ell).\]
\end{lemma}

\begin{proof}
Write $g_1=\gamma_1 g_1'$. 
Then $\gcd(g_1',x^m-1)=1$ and there exists $u$ such that $ug_1'\equiv 1 \pmod{x^m-1}$. Therefore, $ug_1=\gamma_1$ and 
\[
\langle(g_1,\ldots,g_\ell)\rangle
=
\langle(ug_1,\ldots,ug_\ell)\rangle
=
\langle(\gamma_1,ug_2,\ldots,ug_\ell)\rangle.
\]
After relabeling the generators, we may
assume that $g_1\mid x^m-1$ without loss of generality.
\end{proof}
 
For a codeword $c=(bg_1,\dots,bg_\ell)$, let
\[
\bsupp(c):=\{j:bg_j\neq 0\}\subseteq\{1,\dots,\ell\}
\]
be its \emph{block support}. Then each nonzero codeword of $C$ has either full or proper block support and belongs to the following two subsets of $C$: 
\[
C^{\mathrm{full}}:=\{c\in C : \bsupp(c)=\{1,\dots,\ell\}\},
\]
\[
C^{\mathrm{prop}}:=\{c\in C\setminus\{\mathbf 0\} :
\bsupp(c)\subsetneq\{1,\dots,\ell\}\}.
\]
We then have
\begin{equation} \label{trivialbound}
d(C)=\min\Big\{
\min_{c\in C^{\mathrm{full}}}\wt(c),\ \
\min_{c\in C^{\mathrm{prop}}}\wt(c)\Big\}.
\end{equation}
\subsection{Block-support bound from block-shortened and block-punctured codes}
We now introduce two operations on $C$ at the level of quasi-cyclic blocks.

For a subset $U\subseteq\{1,\dots,\ell\}$, the \emph{(block-)punctured code}
\[
C_U:=\big\langle(g_j)_{j\in U}\big\rangle\ \subseteq R^{|U|}
\]
is obtained by deleting the blocks outside $U$. The code $C_U$ is a one-generator code of
index $|U|$ and length $|U|m$.

For a block $j$, the \emph{(block-)shortened code}
\[
\widetilde C_j:= \big\langle(\eta_j g_1,\dots,\eta_j g_{j-1},
\eta_j g_{j+1},\dots,\eta_j g_\ell)\big\rangle\ \subseteq R^{\ell-1}
\]
is obtained by restricting to the codewords of $C$ whose $j$-th block is zero
and deleting that block. The code $\widetilde C_j$ is a one-generator code of index $\ell-1$ and length $(\ell-1)m$. We
adopt the convention $d(\{\mathbf 0\})=+\infty$ when $\widetilde C_j$ is zero.

\begin{lemma}\label{vanish}
Let $(bg_1,\dots,bg_\ell)$ be a codeword of $C$, for some $b\in R$. 
For each  $j$, the block $ bg_j\equiv 0 \pmod{x^m-1}$ if and only if
$\eta_j\mid b$.
\end{lemma}
\begin{proof}
Write $g_j=\gamma_j g_j'$ and $x^m-1=\gamma_j\eta_j$, so that
$\gcd(g_j',\eta_j)=1$. Then \[bg_j\equiv 0\pmod{x^m-1}\] means
$\eta_j\mid bg_j'$, which is equivalent to $\eta_j\mid b$.
\end{proof}

 We now bound the first term  in Equation \eqref{trivialbound} with punctured codes and the second term with shortened codes in the  following two lemmas.

 \begin{lemma}\label{puncturedbound} Let $c$ be a codeword in $C^{\mathrm{full}}$. Fix $1 \le s \le \ell-1$. Then
 \[
\wt(c) \ge \Big\lceil \binom{\ell-1}{s-1}^{-1}\!\!\sum_{|U|=s} d(C_U)\Big\rceil.
 \]
\end{lemma}

\begin{proof} Write $c=(bg_1,\dots,bg_\ell)$. 
Let $w_j:= \wt(bg_j)$, so that $\sum_j w_j = \wt(c)$. 
Since  the codeword $c$ has full support, we have $\sum_{j\in U}w_j\ge d(C_U)$ for every subset $U$ of size $s$.
Summing over the $\binom{\ell}{s}$ subsets of size $s$, each block $j$ is
counted in $\binom{\ell-1}{s-1}$ of them, so
\[
\sum_{|U|=s} d(C_U)\ \le\ \sum_{|U|=s}\sum_{j\in U}w_j
=\binom{\ell-1}{s-1}\sum_{j=1}^{\ell}w_j
=\binom{\ell-1}{s-1}\wt(c).
\]
Dividing by $\binom{\ell-1}{s-1}$ and rounding up gives the bound.
\end{proof}

 \begin{lemma}\label{shortenedbound} Let $c$ be a codeword in $C^{\mathrm{prop}}$. Then
\[
\wt(c) \ge \min_{1\le j\le \ell} d\big(\widetilde C_j\big).
\]
\end{lemma}

\begin{proof} Write $c=(bg_1,\dots,bg_\ell)$.  Since $c$ does not have full support, it has some block $j$ which is zero. 
By Lemma~\ref{vanish}, $\eta_j\mid b$ and so $b=a\eta_j$ for some $a \in R$. Then $c$ has the form
\[
c=(a\eta_j g_1,\dots,a\eta_j g_{j-1},0,a\eta_j g_{j+1},\dots,a\eta_j g_\ell)
\]
and
$\wt(c )\ge d(\widetilde C_j)$. 
\end{proof}

For convenience, we will define the bounds from Lemmas~\ref{puncturedbound}
and~\ref{shortenedbound} as follows. For each
$1\le s\le\ell-1$, the \emph{$s$-block-puncturing bound} is
\[
\Phi_s(C):=\Big\lceil\tbinom{\ell-1}{s-1}^{-1}\!\!\sum_{|U|=s}d(C_U)\Big\rceil,
\]
and the \emph{block-puncturing bound} is
$\Phi(C):=\max_{1\le s\le\ell-1}\Phi_s(C)$. The
\emph{block-shortening bound} is
\[
\Psi(C):=\min_{1\le j\le\ell} d(\widetilde C_j).
\]
We then have the following block-support bound for one-generator quasi-cyclic codes.  

\begin{theorem}\label{bsuppbound} Let $C$ be a  one-generator quasi-cyclic code of index $\ell$. Then
\begin{align*}
d(C)\ \ge\ d_{\mathrm{bsupp}}(C)&:= \min \{ \Phi(C),\Psi(C) \}\\
&=\min \Big\{ \max_{1 \le s \le \ell-1} \Big\{\Big\lceil \binom{\ell-1}{s-1}^{-1}\!\!\sum_{|U|=s}
d(C_U)\Big\rceil \Big\}
, 
\min_{1\le j\le \ell} d\big(\widetilde C_j\big) \Big\}.
\end{align*}
\end{theorem}

\begin{proof} This follows from Lemmas \ref{puncturedbound}, \ref{shortenedbound}, and  Equation \eqref{trivialbound}.
\end{proof}

\begin{example}\label{ex1}
Let $\mathbb{F}_4=\{0,1,\omega,\omega^2\}$ and $m=5$, so $R=\mathbb{F}_4[x]/(x^5-1)$.
Consider the one-generator quasi-cyclic code
$C=\langle(g_1,g_2,g_3,g_4)\rangle\subseteq R^4$ of index $\ell=4$ with
\[
g_1=1,\qquad
g_2=\omega x^4+x^3+x^2,\qquad
g_3=x^4+x^3+x^2+x,\qquad
g_4=\omega^2x^4+\omega x^3+\omega^2x^2+x.
\]
Since $g_1=1$ we have $g=\gcd(g_1,\dots,g_4,x^5-1)=1$. This code $C$ has parameters $[20,5,12]_4$.
To demonstrate the block-support bound in Theorem \ref{bsuppbound}, we show the calculations for the block-puncturing and block-shortening bounds using MAGMA. 

\bigskip
\noindent\emph{Block-puncturing bound.} We first compute the distance of each block-punctured code in the following Table \ref{tab:ex1}. 
\begin{table}[htbp]
\centering
\small
\begin{tabular}{@{}c@{\qquad}c@{\qquad}c@{}}

\begin{tabular}[t]{c|c}
\hline
\(U\) & \(d(C_U)\) \\
\hline
\(\{1\}\) & \(1\) \\
\(\{2\}\) & \(1\) \\
\(\{3\}\) & \(2\) \\
\(\{4\}\) & \(1\) \\
\hline
\end{tabular}
&
\begin{tabular}[t]{c|c}
\hline
\(U\) & \(d(C_U)\) \\
\hline
\(\{1,2\}\) & \(4\) \\
\(\{1,3\}\) & \(4\) \\
\(\{1,4\}\) & \(5\) \\
\(\{2,3\}\) & \(4\) \\
\(\{2,4\}\) & \(4\) \\
\(\{3,4\}\) & \(4\) \\
\hline
\end{tabular}
&
\begin{tabular}[t]{c|c}
\hline
\(U\) & \(d(C_U)\) \\
\hline
\(\{1,2,3\}\) & \(7\) \\
\(\{1,2,4\}\) & \(8\) \\
\(\{1,3,4\}\) & \(8\) \\
\(\{2,3,4\}\) & \(7\) \\
\hline
\end{tabular}

\end{tabular}
\caption{Distances of the block-punctured codes \(C_U\) for the
\([20,5,12]_4\) one-generator code of index 4.}
\label{tab:ex1}
\end{table}

This gives us the three $s$-block-puncturing bounds
\[
\Phi_1(C)=\lceil 5/1\rceil=5,\qquad
\Phi_2(C)=\lceil 25/3\rceil=9,\qquad
\Phi_3(C)=\lceil 30/3\rceil=10,
\]
and so the block-puncturing bound is $\Phi(C)= \max_{1\le s\le\ell-1}\Phi_s(C) = 10$. 

\bigskip
\noindent\emph{Block-shortening bound.} 
For $j\in\{1,2,4\}$ we have $\gamma_j=\gcd(g_j,x^5-1)=1$, so
$\eta_j\equiv 0\pmod{x^5-1}$, $\widetilde{C}_j=\{\mathbf 0\}$, and
$d(\widetilde{C}_j)=+\infty$. 
For $j=3$, we have $\gamma_3=\gcd(g_3,x^5+1)=x+1,$
and so
\[
\eta_3=\frac{x^5+1}{x+1} =x^4+x^3+x^2+x+1=:h.
\]
Then 
\[
\widetilde C_3
=
\left\langle
(hg_1,hg_2,hg_4)
\right\rangle
\subseteq R^3,
\]
and  $d(\widetilde{C}_3)=15$. Hence,
$\Psi(C)=\min_{1\le j\le\ell}d(\widetilde{C}_j)=15$.

\bigskip
\noindent\emph{Block-support bound.} Combining the two terms, the block-support bound is
\[
d_{\mathrm{bsupp}}(C)=\min\{\Phi(C),\Psi(C)\}=\min\{10,15\}=10,
\] 
and the code $C$ has actual distance $d(C) =12 >d_{\mathrm{bsupp}}(C)$. 
\end{example}

\subsection{Bounds for  balanced one-generator quasi-cyclic codes}  

We keep the notations same as in the beginning of the section. We now specialize the block-support bound to a class of one-generator
quasi-cyclic codes.

\begin{definition}
We say that $C$ is \emph{balanced} if $\gamma_1=\dots=\gamma_\ell=g$.
\end{definition} 

\begin{remark}\label{guan}
As mentioned in the introduction, our motivation to consider balanced codes comes from the proof of 
\cite[Theorem~1]{guan2023b}. 
There, a one-generator quasi-cyclic code is presented as  $C=\langle(g\mathfrak{f}_1,\dots,g\mathfrak{f}_\ell)\rangle$
with $g\mid x^m-1$, and the argument requires the condition
\begin{equation}\tag{$*$}\label{star}
\gcd\!\Big(\mathfrak{f}_j,\frac{x^m-1}{g}\Big)=1
\quad\text{for every } j.
\end{equation} 
In our situation,  we assume $\gcd(q,m)=1$ so that $x^m-1$ is square free. Then
\[
\gamma_j=\gcd (g\mathfrak{f}_j,\,x^m-1) =g\gcd\Big(\mathfrak{f}_j,\frac{x^m-1}{g}\Big),
\]
which means Condition \eqref{star} holds if and only if $C$ is balanced.
\end{remark}

\begin{theorem}\label{balanced}
The following are equivalent.
\begin{enumerate}
\item[\textup{(i)}] $C$ is balanced.
\item[\textup{(ii)}] $C^{\mathrm{prop}}=\emptyset$.
\item[\textup{(iii)}] $\widetilde{C}_j=\{\mathbf{0}\}$ for every $j\in\{1,\ldots,\ell\}$.
\item[\textup{(iv)}] $\Psi(C)=+\infty$.
\end{enumerate}
\end{theorem}
 
\begin{proof} 
\noindent\textup{(i)}$\Rightarrow$\textup{(ii)}.
Since $C$ is balanced, $\gamma_1=\dots=\gamma_\ell=g$, and
\[
\eta_1=\dots=\eta_\ell=  \dfrac{x^m-1}{g}=:\eta. 
\] 
Let $c=(bg_1,\dots,bg_\ell)$ be a codeword. By Lemma~\ref{vanish}, we have $bg_j\equiv 0$ if and only if $\eta\mid b$. Then either $c$ is a zero codeword or $c$ has full support. 
This shows that $C^{\mathrm{prop}}=\emptyset$.
 
\smallskip
\noindent\textup{(ii)}$\Rightarrow$\textup{(i)}.
Fix $j$ and consider $c=\eta_j(g_1,\ldots,g_\ell)\in C$. 
Writing $g_j=\gamma_j g_j'$, we have 
\[
\eta_j g_j  = \eta_j \gamma_j g_j' \equiv 0 \pmod{x^m-1},
\]
and since we are assuming $C^{\mathrm{prop}}=\emptyset$, it follows that $c$ is the
zero codeword, that is, $\eta_j g_1\equiv\dots\equiv\eta_j g_\ell\equiv 0 \pmod{x^m-1}$.
By Lemma \ref{vanish}, $\eta_j$ is divisible by each of $\eta_1,\dots,\eta_\ell$, and hence by
$\eta:=\text{lcm}(\eta_1,\dots,\eta_\ell)$.
As $\eta_j \mid \eta$, we get $\eta_j=\eta$.
Since $j$ was arbitrary, $\eta_1=\dots=\eta_\ell,$
 which implies $\gam_1=\dots=\gam_\ell,$ and so $C$ is balanced.
 
\smallskip
\noindent\textup{(ii)}$\Leftrightarrow$\textup{(iii)}.
There is a nonzero codeword $c$ of $C$ with
$j\notin\mathrm{bsupp}(c)$ if and only if $\widetilde{C}_j\neq\{\mathbf{0}\}$, and
$C^{\mathrm{prop}}\neq\emptyset$ if and only if this occurs for some $j$.
 
\smallskip
\noindent\textup{(iii)}$\Leftrightarrow$\textup{(iv)}.
This follows from the definition $\Psi(C)=\min_{1\le j\le\ell} d(\widetilde{C}_j)$ and the convention
$d(\{\mathbf{0}\})=+\infty$.
\end{proof}

\begin{theorem}\label{monotone}
Let $C$ be a balanced one-generator quasi-cyclic code of index $\ell$. Then
\[
\ell d(\langle g \rangle) = \Phi_1(C)\ \le\ \Phi_2(C)\ \le\ \cdots\ \le\ \Phi_{\ell-1}(C),
\]
  and consequently $\Phi(C)=\Phi_{\ell-1}(C)$.
\end{theorem}
\begin{proof}
Let 
\[
S_s=\sum_{|U|=s}d(C_U), \qquad A_s=\binom{\ell-1}{s-1}^{-1}S_s, 
\]
so that
$\Phi_s(C)=\lceil A_s\rceil$. We will prove that $A_s\le A_{s+1}$ for $1\le s\le\ell-2$.

Fix a subset $W\subseteq\{1,\dots,\ell\}$ with $|W|=s+1$, and let $c=(bg_j)_{j\in W}$ be a nonzero minimum-weight codeword of $C_W$, so that $d(C_W)=\sum_{j\in W}\wt(bg_j)$. 
For each  subset $U\subset W$ of size $s$,  the tuple $(bg_j)_{j\in U}$ is a codeword of $C_U$, which gives $d(C_U)\le\sum_{j\in U}\wt(bg_j)$. 
Summing over the $s+1$ such subsets $U$, with each block $j$ of $C_W$ being counted $s$ times, we get
\[
\sum_{\substack{U\subset W\\ |U|=s}}d(C_U)\ \le\ s\sum_{j\in W}\wt(bg_j)
\ =\ sd(C_W).
\]
Summing this over all $W$ of size $s+1$, where each  subset $U$ is counted
$\ell-s$ times,  we get $(\ell-s)S_s\le sS_{s+1}$. Since
$\binom{\ell-1}{s}=\frac{\ell-s}{s}\binom{\ell-1}{s-1}$, this gives
$A_s\le A_{s+1}$ and consequently $\Phi_s(C) \le \Phi_{s+1}(C)$.

 Finally, when $s=1$, each $C_U$ is $\langle g \rangle$ and so 
$\Phi_1(C)= \ell d(\langle g \rangle)$. 
\end{proof}

\begin{example}\label{ex2}
Let $\mathbb{F}_4=\{0,1,\omega,\omega^2\}$ and $m=5$, so
$R=\mathbb{F}_4[x]/(x^5-1)$.
Consider the one-generator quasi-cyclic code
$C=\langle(g_1,g_2,g_3,g_4)\rangle\subseteq R^4$ of index $\ell=4$ with
\[
g_1=x+1,\qquad
g_2=(x+1)^2,\qquad
g_3=(x+1)(x+\omega)^2,\qquad
g_4=(x+1)^2(x^2+\omega^2x+\omega^2).
\]
For every $j\in\{1,2,3,4\}$, we have
$\gamma_j=\gcd(g_j,x^5-1)=x+1$. Hence,
\[
g=\gcd(g_1,g_2,g_3,g_4,x^5-1)=x+1,
\]
so $C$ is balanced. This code $C$ has parameters $[20,4,13]_4$.
To demonstrate the bound for balanced codes, we show the calculations
for the block-puncturing and block-shortening bounds using  \textsc{Magma}.

\bigskip
\noindent\emph{Block-puncturing bound.}
We first compute the distance of each block-punctured code in the
following Table~\ref{tab:ex2}.
\begin{table}[htbp]
\centering
\small
\begin{tabular}{@{}c@{\qquad}c@{\qquad}c@{}}

\begin{tabular}[t]{c|c}
\hline
\(U\) & \(d(C_U)\) \\
\hline
\(\{1\}\) & \(2\) \\
\(\{2\}\) & \(2\) \\
\(\{3\}\) & \(2\) \\
\(\{4\}\) & \(2\) \\
\hline
\end{tabular}
&
\begin{tabular}[t]{c|c}
\hline
\(U\) & \(d(C_U)\) \\
\hline
\(\{1,2\}\) & \(4\) \\
\(\{1,3\}\) & \(5\) \\
\(\{1,4\}\) & \(6\) \\
\(\{2,3\}\) & \(6\) \\
\(\{2,4\}\) & \(5\) \\
\(\{3,4\}\) & \(4\) \\
\hline
\end{tabular}
&
\begin{tabular}[t]{c|c}
\hline
\(U\) & \(d(C_U)\) \\
\hline
\(\{1,2,3\}\) & \(8\) \\
\(\{1,2,4\}\) & \(9\) \\
\(\{1,3,4\}\) & \(9\) \\
\(\{2,3,4\}\) & \(8\) \\
\hline
\end{tabular}

\end{tabular}
\caption{Distances of the block-punctured codes $C_U$ for the
$[20,4,13]_4$ balanced one-generator code of index 4.}
\label{tab:ex2}
\end{table}

This gives us the three $s$-block-puncturing bounds
\[
\Phi_1(C)=\lceil 8/1\rceil=8,\qquad
\Phi_2(C)=\lceil 30/3\rceil=10,\qquad
\Phi_3(C)=\lceil 34/3\rceil=12.
\]
Thus
\[
\Phi_1(C)=8<\Phi_2(C)=10<\Phi_3(C)=12,
\]
so all the inequalities in Theorem~\ref{monotone} are strict, and the
block-puncturing bound is $\Phi(C)=\Phi_3(C)=12$.

\bigskip
\noindent\emph{Block-shortening bound.}
Since $C$ is balanced, Theorem~\ref{balanced} gives
$\widetilde C_j=\{\mathbf 0\}$ and $d(\widetilde C_j)=+\infty$ for every
$j\in\{1,2,3,4\}$. Hence,
\[
\Psi(C)=\min_{1\le j\le\ell}d(\widetilde C_j)=+\infty.
\]

\bigskip
\noindent\emph{Block-support bound.}
Combining the two terms, the block-support bound is
\[
d_{\mathrm{bsupp}}(C)=\min\{\Phi(C),\Psi(C)\}
=\min\{12,+\infty\}=12,
\]
and the code $C$ has actual distance
$d(C)=13>d_{\mathrm{bsupp}}(C)$.
\end{example}

\subsection{The Jensen bound for balanced codes}

We follow the setting of the Jensen bound in Section 2.3.
Since $C=\langle(g_1,\dots,g_\ell)\rangle$ is one-generator,
Equation \eqref{spanconstituent} shows that the constituent $C_i$ is the span of
$\bigl(g_1(\xi^{u_i}),\dots,g_\ell(\xi^{u_i})\bigr)$ over $E_i$, and is
therefore at most one-dimensional.  It is nonzero if and only if
$f_i\nmid g$, so the index set of Theorem~\ref{jensen} is
\[
  I(C)=\{i : f_i\nmid g\}=\{i_1,\dots,i_t\},
  \qquad t=\nu-\#\{i : f_i\mid g\}.
\]
Order the indices so that $d(C_{i_1})\le\dots\le d(C_{i_t})$ and write
$D_z:=C_{i_z}$ for $1\le z\le t$.
Then Theorem~\ref{jensen} gives
\[
  d(C)\ \ge\ d_J(C)\ =\ \min_{1\le z\le t}
  \Bigl\{d(D_z)\,
  d\bigl(\langle\theta_{i_1}\rangle\oplus\dots\oplus\langle\theta_{i_z}\rangle\bigr)
  \Bigr\}.
\]

\begin{theorem}\label{jensen-balanced}
Let $C$ be a balanced one-generator quasi-cyclic code of index $\ell$. Then
\[
d_J(C) =\ \ell d(\langle g\rangle).
\]
\end{theorem}

\begin{proof} 

1. For $1\le z\le t$, we let $\beta_z:=\xi^{u_{i_z}}$.
Then $D_z$ is spanned by
$\bigl(g_1(\beta_z),\dots,g_\ell(\beta_z)\bigr)$ and hence
\[
d(D_z)=\#\{j: g_j(\beta_z)\neq 0\}.
\]
For any root $\beta$ of $x^m-1$ and any $j$, we have $g_j(\beta)=0$ if and only if
$\gamma_j(\beta)=0$. Since $C$ is balanced, $\gamma_1=\dots=\gamma_\ell=g$, and this
holds if and only if $g(\beta)=0$.
As $D_z\neq\{\mathbf 0\}$,
we have $g(\beta_z)\neq 0$, hence $g_j(\beta_z)\neq 0$ for every $j$ and
$d(D_z)=\ell$. In other words, each $D_z$ is an $[\ell,1,\ell]$ code.

2. We have
\[
\langle\theta_{i_1}\rangle\ \subseteq\ \langle\theta_{i_1}\rangle\oplus\langle\theta_{i_2}\rangle
\ \subseteq\ \dots\ \subseteq\
\langle\theta_{i_1}\rangle\oplus\dots\oplus\langle\theta_{i_t}\rangle.
\]
Also, for $1 \le z \le t$, we have
\[
\langle\theta_{i_1}\rangle \oplus \dots \oplus \langle\theta_{i_z}\rangle
=\Big\langle\tfrac{x^m-1}{f_{i_1}}\Big\rangle \oplus\dots \oplus\Big\langle\tfrac{x^m-1}{f_{i_z}}\Big\rangle
=\Big\langle\gcd\Big(\tfrac{x^m-1}{f_{i_1}},\dots,\tfrac{x^m-1}{f_{i_z}}\Big)\Big\rangle
=\Big\langle\tfrac{x^m-1}{f_{i_1}\cdots f_{i_z}}\Big\rangle.
\]
Then
\[
d(\langle\theta_{i_1}\rangle)\ \ge\ d(\langle\theta_{i_1}\rangle\oplus\langle\theta_{i_2}\rangle)
\ \ge\ \dots\ \ge\
d(\langle\theta_{i_1}\rangle\oplus\dots\oplus\langle\theta_{i_t}\rangle)=d(\langle g\rangle).
\]
Combining parts 1 and 2 gives the equality $d_J(C) = \ell d(\langle g \rangle)$. 
\end{proof}

\begin{corollary}\label{jensen-chain}
Let $C$ be a balanced one-generator quasi-cyclic code of index $\ell$. 
Then
\[
d_J(C)=\Phi_1(C)\ \le\ \Phi_2(C)\ \le\ \cdots\ \le\ \Phi_{\ell-1}(C)
=d_{\mathrm{bsupp}}(C)\ \le\ d(C).
\]
In particular, the block-support bound is generally better than the Jensen bound.
\end{corollary}

\begin{proof}
Theorem~\ref{jensen-balanced} and Theorem~\ref{monotone} give
$d_J(C)=\ell\,d(\langle g\rangle)=\Phi_1(C)$ together with the increasing chain up to $\Phi_{\ell-1}(C)=\Phi(C)$.  
By Theorem~\ref{balanced}, $\Psi(C)=+\infty$, so
$d_{\mathrm{bsupp}}(C)=\min\{\Phi(C),\Psi(C)\}=\Phi(C)$.  
The last inequality is
Theorem~\ref{bsuppbound}.
\end{proof}

\begin{example}\label{ex3}
Consider again the balanced one-generator quasi-cyclic code $C$ with parameters $[20,4,13]_4$ in Example~\ref{ex2}. Recall that $\ell=4$, $m=5$, and
$g=x+1$.

\bigskip
\noindent\emph{Jensen bound.}
Over $\mathbb{F}_4$, we have
\[
x^5-1=f_1(x)f_2(x)f_3(x),
\]
where
\[
f_1(x):=x+1,
\qquad
f_2(x):=x^2+\omega x+1,
\qquad
f_3(x):=x^2+\omega^2x+1,
\]
so that $\nu=3$. Since $g=x+1=f_1$, we have $f_1\mid g$ and
$f_2,f_3\nmid g$, so
\[
I(C)=\{i: f_i\nmid g\}=\{2,3\},
\qquad t=2 .
\]
Both nonzero constituents have distance $4$, so  we take $i_1=2$, $i_2=3$, giving $D_1=C_2$ and $D_2=C_3$.
Since $C$ is balanced, Theorem~\ref{jensen-balanced} shows that $D_1$ and
$D_2$ are $[4,1,4]$ codes over their respective constituent fields.
The minimal cyclic codes $\langle\theta_2\rangle$ and
$\langle\theta_3\rangle$, whose check polynomials are $f_2$ and $f_3$,
respectively, are $[5,2,4]_4$ codes. Moreover,
\[
\langle\theta_2\rangle\oplus\langle\theta_3\rangle
=\Big\langle\tfrac{x^5-1}{f_2f_3}\Big\rangle
=\langle x+1\rangle
=\langle g\rangle
\]
is a $[5,4,2]_4$ cyclic code. Hence, the Jensen bound is
\[
\begin{aligned}
d_J(C)
&=\min\Big\{
 d(D_1)d(\langle\theta_2\rangle),
 d(D_2)d(\langle\theta_2\rangle\oplus\langle\theta_3\rangle)
 \Big\}\\
&=\min\{4\cdot4,\,4\cdot2\}=8.
\end{aligned}
\]
Equivalently, since $d(\langle g\rangle)=d(\langle x+1\rangle)=2$,
Theorem~\ref{jensen-balanced} gives
\[
d_J(C)=\ell d(\langle g\rangle)=4\cdot2=8.
\]

\bigskip
\noindent\emph{Comparison with the block-support bound.}
In Example~\ref{ex2}, we obtained
\[
\Phi_1(C)=8,
\qquad
\Phi_2(C)=10,
\qquad
\Phi_3(C)=12,
\]
and $d_{\mathrm{bsupp}}(C)=12$, whereas $d(C)=13$. Therefore,
\[
d_J(C)=\Phi_1(C)=8
<\Phi_2(C)=10
<\Phi_3(C)=d_{\mathrm{bsupp}}(C)=12
<d(C)=13.
\]
The one-block puncturing bound recovers the Jensen bound, while the
two-block and three-block puncturing bounds give  improvements.
\end{example}

\section{Hermitian LCD one-generator quasi-cyclic codes}

We now demonstrate the applicability of the bounds from
Section~3 in the search for Hermitian LCD one-generator quasi-cyclic codes. 
To this end, we first derive a polynomial description of the Hermitian hull and a characterization of the Hermitian LCD property. 
In particular, these results strengthen the corresponding results of
\cite{guneri2023}.
We then present quaternary Hermitian LCD one-generator quasi-cyclic codes obtained by computer search and compare their parameters with known results in the literature. 
For the search, we used the algebraic criteria to identify Hermitian LCD candidates, and  we used the distance bounds to screen candidates with potentially high distances before their exact minimum distances are computed.
We record the block-support bounds for some balanced and unbalanced examples from these constructions, illustrating their effectiveness as tools for the search of one-generator quasi-cyclic codes with good parameters.  

\subsection{Hermitian hulls and the LCD criterion}
Throughout this section, we let $q=q^2_0$. Since $q=q_0^2$ is a square, we can define the \textit{Hermitian inner product}
\[
\langle \mathbf{x}, \mathbf{y} \rangle_h := \sum_{i=1}^n x_i y_i^{q_0}.
\]
The \textit{Hermitian dual} of a linear code $C\subseteq F^n$ is
\[
C^{\perp_h}:=\{\mathbf{x}\in F^n : \langle \mathbf{x},\mathbf{c}\rangle_h=0
\text{ for all } \mathbf{c}\in C\},
\]
and the \textit{Hermitian hull} of $C$ is
$\mathrm{Hull}_h(C):=C\cap C^{\perp_h}$. If
$\mathrm{Hull}_h(C)=\{\mathbf{0}\}$, then $C$ is called \textit{Hermitian
LCD}.

For a polynomial $f(x)=\sum_{i} f_ix^i\in F[x]$, its \textit{conjugate} is defined as
$f^{[q_0]}(x)=\sum_{i} f_i^{q_0}x^i$ and its
\textit{conjugate-reciprocal} is defined as
\[
f^{\dagger}(x)=x^{\deg f(x)}f^{[q_0]}(x^{-1}).
\]
We say a polynomial $f$ is \textit{self-conjugate-reciprocal} if
$f^{\dagger}(x)=\alpha f(x)$ for some $\alpha\in F$. For a polynomial
$f(x)=\sum_{i=0}^{m}f_ix^i$ of degree at most $m$, we define its
\textit{conjugate transpose polynomial} as
\[
\widehat f(x)=x^{m}f^{[q_0]}(x^{-1})=\sum_{i=0}^{m}f_{m-i}^{q_0}x^i .
\]
Then $\widehat f(x)=x^{m-\deg f(x)}f^{\dagger}(x)$. 

We keep the notation of Section~3 and write
\[
\mathbf g:=(g_1,\dots,g_\ell),
\]
so that $C=\langle\mathbf g\rangle$. By Lemma~\ref{normalization} we
assume throughout this section that $g_1\mid x^m-1$. Then
$g=\gcd(g_1,\dots,g_\ell)$ and $h=(x^m-1)/g$.

Under the identification in Section~2.1, the Hermitian inner product
carries over to $R$ and $R^\ell$.  For
$u=\sum_{i=0}^{m-1}u_ix^i$, $v=\sum_{i=0}^{m-1}v_ix^i \in R$ and
$\mathbf a=(a_1,\dots,a_\ell)$, $\mathbf b=(b_1,\dots,b_\ell) \in R^\ell$,
we define
\[
\langle u,v\rangle_h:=\sum_{i=0}^{m-1}u_i\,v_i^{q_0},
\qquad
\langle\mathbf a,\mathbf b\rangle_h
:=\sum_{j=1}^{\ell}\langle a_j,b_j\rangle_h .
\]

We have the following lemma, cf. \cite[Proposition~2]{galindo2018} and \cite[Lemma~1]{lv2020}.
\begin{lemma}\label{generating}
For $u,v\in R$,
\[
u\,\widehat{v}
=\sum_{\tau=0}^{m-1}\big\langle u,\,x^{\tau}v\big\rangle_h\,x^{\tau}.
\]
\end{lemma} 
\begin{proof}
Reading subscripts modulo $m$, we have
$\widehat{v}=\sum_{i=0}^{m-1}v_{-i}^{\,q_0}x^{i}$ in $R$. 
For any $u,w\in R$, the coefficient of $x^{\tau}$ in $uw$ is
$\sum_{i=0}^{m-1}u_i\,w_{\tau-i}$. 
Taking $w=\widehat{v}$, so that
$w_{\tau-i}=v_{i-\tau}^{\,q_0}$, the coefficient of $x^{\tau}$ in
$u\widehat{v}$ is
\[
\sum_{i=0}^{m-1}u_i\,v_{i-\tau}^{\,q_0}
=\sum_{i=0}^{m-1}u_i\,(x^{\tau}v)_i^{\,q_0}
=\big\langle u,\,x^{\tau}v\big\rangle_h.  
\]
This completes the proof. 
\end{proof}

Let
\[
\Lambda(x):=\sum_{j=1}^{\ell} g_j(x)\widehat{g_j}(x) \in F[x],
\qquad
\lambda(x):=\frac{x^m-1}{\gcd\big(\Lambda(x),x^m-1\big)} \in F[x].
\]
For convenience, from now on  we will still use the notations $\Lambda$ and $\lambda$ to denote their representatives in $R$ where appropriate.
\begin{theorem}\label{hull}
The Hermitian hull of $C$ is the one-generator quasi-cyclic code
\[
\mathrm{Hull}_h(C)
=\{b\mathbf g : b \Lambda=0\}
=\langle \lambda \mathbf g\rangle .
\]
In particular,
\[
\dim  \mathrm{Hull}_h(C)=\deg\gcd\big(\Lambda(x),h(x)\big) =m-\deg\lambda-\deg g. 
\]
\end{theorem}

\begin{proof}
1. We first show that, for $\mathbf a=(a_1,\dots,a_\ell)\in R^\ell$,
\begin{equation}\label{herm-dual}
\mathbf a\in C^{\perp_h}\iff
\sum_{j=1}^{\ell} a_j \widehat{g_j}=0.
\end{equation}
Applying Lemma~\ref{generating} for each $j$ and summing, we have
\[
\sum_{j=1}^{\ell}a_j\widehat{g_j}
=\sum_{\tau=0}^{m-1}\Big(\sum_{j=1}^{\ell}
\big\langle a_j,\,x^{\tau}g_j\big\rangle_h\Big)x^{\tau}
=\sum_{\tau=0}^{m-1}
\langle\mathbf a,\,x^{\tau}\mathbf g\rangle_h\,x^{\tau}.
\]
On the other hand, $C$ is spanned by the codewords $x^{\tau}\mathbf g$,
$0\le \tau\le m-1$,  and  $\mathbf a\in C^{\perp_h}$ if and only if 
$\langle\mathbf a, x^{\tau}\mathbf g\rangle_h = 0$ for all $\tau$. 
Hence $\sum_{j=1}^{\ell} a_j \widehat{g_j}=0$ if and only if
$\mathbf a\in C^{\perp_h}$, which proves \eqref{herm-dual}.

2. Every codeword of $C$ is $b \mathbf g=(bg_1,\dots,bg_\ell)$ for some
$b\in R$, and $\sum_j bg_j\widehat{g_j}=b\Lambda$. Hence, 
\eqref{herm-dual} gives
\[
\mathrm{Hull}_h(C)=C\cap C^{\perp_h}
=\{ b \mathbf g : b \Lambda=0\}.
\]
Write $\Lambda_0:=\gcd(\Lambda,x^m-1)$ and
$\Lambda=\Lambda_0\Lambda'$, so that $x^m-1=\Lambda_0\lambda$. Then $\gcd(\Lambda',\lambda)=1$ and
\[
b \Lambda = 0
\iff x^m-1 \mid b \Lambda 
\iff \lambda\mid b\Lambda'
\iff \lambda\mid b,
\]
so $\mathrm{Hull}_h(C)=\{b\mathbf g : \lambda\mid b\}
=\langle \lambda\mathbf g\rangle$.

3. Since $g\mid g_j$ for every $j$, we have $g\mid\Lambda$, and so
$g\mid\Lambda_0$. Since $\lambda=(x^m-1)/\Lambda_0$ and $h=(x^m-1)/g$,
we have $\lambda \mid h$. Consequently, $\lambda g\mid x^m-1$. As
$g=\gcd(g_1,\dots,g_\ell)$, we obtain
\[
\gcd(\lambda g_1,\dots,\lambda g_\ell,x^m-1)
=\gcd(\lambda g,\,x^m-1)=\lambda g,
\]
and $\dim\mathrm{Hull}_h(C)=m-\deg\lambda-\deg g$.

4. Finally, since $h\mid x^m-1$, any common divisor of $\Lambda$ and $h$
divides $\Lambda_0$, so $\gcd(\Lambda,h)=\gcd(\Lambda_0,h)$. As
$\gcd(q,m)=1$ gives $\gcd(g,h)=1$, and $g\mid\Lambda_0\mid gh$, we get
$\Lambda_0=g\cdot\gcd(\Lambda_0,h)$. Hence
$\deg\gcd(\Lambda,h)=\deg\Lambda_0-\deg g=m-\deg\lambda-\deg g$.
\end{proof}

\begin{corollary}\label{lcd}
$C$ is Hermitian LCD if and only if $\gcd\big(\Lambda(x),h(x)\big)=1$.
\end{corollary}

To aid the search for Hermitian LCD codes with good parameters, we use the following two propositions. 

\begin{proposition}\label{scr}
If $C$ is Hermitian LCD, then $g$ and $h$ are self-conjugate-reciprocal.
\end{proposition}

\begin{proof}   
Since $\Lambda(x)$ is self-conjugate-reciprocal,  the polynomials $\lambda$ and $\gcd(\Lambda,x^m-1)$ are
also self-conjugate-reciprocal. 
If $C$ is Hermitian LCD, then
$\gcd(\Lambda,h)=1$ by Corollary~\ref{lcd},
and so $g=\gcd(\Lambda,x^m-1)$ by Theorem~\ref{hull}. 
Then $g$ and $h=(x^m-1)/g$ are also self-conjugate-reciprocal.
\end{proof}

For a divisor $a$ of $h(x)$, let
\[
C_a:=aC=\big\langle(ag_1,\dots,ag_\ell)\big\rangle
\]
be a subcode of $C$. 
We note that $\dim C_a=\dim C-\deg a$,  and when $a\neq h$, the
code $C_a$ is a nonzero subcode of $C$ with $d(C_a)\ge d(C)$.

\begin{proposition}\label{subcode}
Let $a$ be a self-conjugate-reciprocal divisor of $h(x)$.
\begin{enumerate}
\item[\textup{(i)}] If $C$ is Hermitian LCD, then $C_a$ is Hermitian LCD.
\item[\textup{(ii)}] If $C_a$ is Hermitian LCD, then $C$ is Hermitian LCD
if and only if $\gcd(\Lambda,a)=1$.
\end{enumerate}
\end{proposition}
\begin{proof}
We have
\[
\gcd(a g_1,\dots,a g_\ell,x^m-1) = ag,
\]
so the   check
polynomial of $C_a$ is $h/a$. 
By Corollary~\ref{lcd}, the code $C_a$ is Hermitian LCD if and only if
$\gcd(\Lambda_{C_a},h/a)=1$, where
\[
\Lambda_{C_a}=\sum_{j=1}^{\ell}(ag_j)\widehat{ag_j}=a\widehat a\Lambda  =\alpha x^{m-\deg a}a^{2}\Lambda ,
\]
 for some $\alpha\in F^{\ast}$.
Here we note that $\gcd(a,h/a)=1$  and  
$\gcd(x,h/a)=1$. Hence,
\[
\gcd\big(\Lambda_{C_a},h/a\big) 
=\gcd\big(\Lambda,h/a\big).
\]
Moreover, since $h=a\cdot(h/a)$ with $\gcd(a,h/a)=1$, we have
\[
\gcd(\Lambda,h)=\gcd(\Lambda,a)\cdot\gcd\big(\Lambda,h/a\big).
\]

(i) If $C$ is Hermitian LCD, then $\gcd(\Lambda,h)=1$ by
Corollary~\ref{lcd}, so $\gcd(\Lambda,h/a)=1$, and $C_a$ is
Hermitian LCD.

(ii) If $C_a$ is Hermitian LCD, then $\gcd(\Lambda,h/a)=1$, so
$\gcd(\Lambda,h)=\gcd(\Lambda,a)$. Then $C$ is Hermitian LCD if and only
if $\gcd(\Lambda,a)=1$.
\end{proof}

\begin{remark}\label{hullissubcode}
In view of  Propositions \ref{scr} and \ref{subcode}, we see that  $\mathrm{Hull}_h(C)$   is the  subcode $C_\lambda$ of $C$. When $C$ is Hermitian LCD, $\lambda=h$ and so
$C_\lambda=\{0\}$.
\end{remark}

\subsection{Constructions and comparison with known parameters}
In \cite{carlet2018}, Carlet et al.\ showed that any linear code over $\mathbb{F}_q$ with $q>4$ is equivalent to a Hermitian LCD code. 
This motivates us to search for quaternary Hermitian LCD codes with good parameters.
We write $[n,k,d]^{H}_{4}$ for an $[n,k,d]$ Hermitian LCD code over
$\mathbb{F}_4$, and for $1\le k\le n$ we let
\[
d^{H}_{4}(n,k):=\max\bigl\{\,d(C)\;:\;C\ \text{is a Hermitian LCD }[n,k]
\text{ code over }\mathbb{F}_4\,\bigr\},
\]
the largest minimum distance of a Hermitian LCD code with the given length $n$ and dimension $k$. 
Let $d_{4}(n,k)$ denote the largest minimum distance of an $[n,k]$ linear code over $\mathbb{F}_4$, without the LCD restriction.
We have $d^{H}_{4}(n,k)\le d_{4}(n,k)$, and the inequality may be strict in view of the result by Carlet et al. \cite{carlet2018}.  The determination of $d^{H}_{4}$ is a problem distinct from that of $d_{4}$.

For comparison, we use Grassl's table \cite{grassl} for the values of
$d_{4}(n,k)$, and the Araya-Harada table \cite[Table~6]{araya2024} for
the values of $d^{H}_{4}(n,k)$ in the range $n\le 30$. We are not aware of any database recording $d^{H}_{4}(n,k)$ for $n>30$.

 The values $d^H_4(n,k)$ were determined in \cite{araya2020,harada2019,lu2015} 
for $k = 1, 2, 3, n-1, n-2$ and $n-3$.
Complete classifications up to equivalence of the optimal codes were
subsequently obtained by Ishizuka \cite{ishizuka2020} for dimension 2
and by Araya and Harada \cite{araya2022} for dimension 3.

\begin{table}[H]
\centering
\begin{tabular}{c c c c c c}
\hline
No. & Index \(\ell\) & \(m\) & Our code & Araya--Harada entry & \(\Delta\) \\
\hline
1 & \(2\) & \(13\) & \([26,12,10]^H_4\) & \(d^H_4(26,12)=9\text{--}11\) & \(+1\) \\
2 & \(2\) & \(13\) & \([26,13,9]^H_4\) & \(d^H_4(26,13)=8\text{--}10\) & \(+1\) \\
3 & \(2\) & \(15\) & \([30,15,10]^H_4\) & \(d^H_4(30,15)=9\text{--}12\) & \(+1\) \\
\hline
4 & \(3\) & \(7\) & \([21,7,11]^H_4\) & \(d^H_4(21,7)=10\text{--}11\) & \(+1\) \\
5 & \(3\) & \(9\) & \([27,7,15]^H_4\) & \(d^H_4(27,7)=14\text{--}16\) & \(+1\) \\
6 & \(3\) & \(9\) & \([27,9,13]^H_4\) & \(d^H_4(27,9)=12\text{--}14\) & \(+1\) \\
\hline
\end{tabular}
\caption{Comparison with the Araya--Harada table.}
\label{tab:AH}
\end{table}

Using the criterion of Section~4.1 and the bounds of Section~3 as a screen, we conducted a computer search in \textsc{Magma} for Hermitian LCD one-generator quasi-cyclic codes with good parameters, up to length $63$. The search produced six codes improving the corresponding lower bounds of the Araya-Harada table, recorded in Table~\ref{tab:AH}. 

The search also produced 13  codes attaining the exact values listed in \cite[Table~6]{araya2024}. These are
$[10,4,6]^H_4$, $[10,5,5]^H_4$, $[14,6,7]^H_4$, $[14,7,6]^H_4$,
$[15,4,9]^H_4$, $[15,5,8]^H_4$, $[18,6,9]^H_4$, $[18,8,8]^H_4$,
$[20,4,13]^H_4$, $[20,5,12]^H_4$, $[21,6,12]^H_4$, $[27,4,18]^H_4$, and $[30,4,20]^H_4$.

In particular, the code $[21,7,11]^H_4$ attains the upper bound of the corresponding interval in that table, and hence establishes $d_4^H(21,7)=11$. 
We also note that the value $d_4^H(22,10)=9$ was established in~\cite{c104} and we obtain a one-generator quasi-cyclic code $[22,10,9]^H_4$ as an alternative construction.

For lengths $n > 30$ beyond the range of \cite[Table~6]{araya2024},  we obtain 11 codes that match the corresponding entry of Grassl's table \cite{grassl}.
The code $[35,4,24]^H_4$ is optimal, so that
$d^{H}_{4}(35,4)=d_{4}(35,4)=24$. The codes $[33,6,20]^H_4$,
$[33,10,16]^H_4$, $[34,17,11]^H_4$, $[38,9,20]^H_4$, $[38,10,19]^H_4$,
$[38,18,12]^H_4$, $[38,19,12]^H_4$, $[42,18,14]^H_4$, $[42,19,14]^H_4$, and $[46,22,14]^H_4$ attain the minimum distance of the best known linear code of the same length and dimension recorded in \cite{grassl}.

Generator polynomials for all codes obtained in the
search are listed in the appendix.

\subsection{Performance of the block-support bounds}

We demonstrate the bounds of Section~3 for some Hermitian LCD
one-generator quasi-cyclic codes constructed above. 
We separate the codes into balanced and unbalanced cases, since the balanced condition eliminates the block-shortening term and gives
\[
d_{\mathrm{bsupp}}(C)=\Phi(C)=\Phi_{\ell-1}(C),
\]
whereas in the unbalanced case both $\Phi(C)$ and $\Psi(C)$ must be
considered.  Tables~\ref{tab:balancedbounddemo} and
\ref{tab:unbalancedbounddemo} summarize the resulting bounds.

\begin{table}[H]
\centering
\small
\renewcommand{\arraystretch}{1.15}
\setlength{\tabcolsep}{4.5pt}
\begin{tabular}{c c c c c c c c c}
\hline
No. & $\ell$ & $m$ & Code $C$
& $\Phi_1(C)=d_J(C)$
& $\Phi_2(C)$
& $\Phi_3(C)$
& $\Phi(C)=d_{\mathrm{bsupp}}(C)$
& $d(C)-d_{\mathrm{bsupp}}(C)$ \\
\hline
1  & $2$ & $5$  & $[10,4,6]^H_4$   & $4$ & --   & --   & $4$  & $2$ \\
2  & $2$ & $7$  & $[14,6,7]^H_4$   & $4$ & --   & --   & $4$  & $3$ \\
3  & $2$ & $9$  & $[18,6,9]^H_4$   & $4$ & --   & --   & $4$  & $5$ \\
4  & $2$ & $11$ & $[22,10,9]^H_4$  & $4$ & --   & --   & $4$  & $5$ \\
5  & $2$ & $13$ & $[26,12,10]^H_4$ & $4$ & --   & --   & $4$  & $6$ \\
\hline
6  & $3$ & $5$  & $[15,4,9]^H_4$   & $6$ & $8$  & --   & $8$  & $1$ \\
7  & $3$ & $5$  & $[15,5,8]^H_4$   & $3$ & $7$  & --   & $7$  & $1$ \\
8  & $3$ & $7$  & $[21,6,12]^H_4$  & $6$ & $9$  & --   & $9$  & $3$ \\
9  & $3$ & $7$  & $[21,7,11]^H_4$  & $3$ & $9$  & --   & $9$  & $2$ \\
10 & $3$ & $9$  & $[27,4,18]^H_4$  & $9$ & $14$ & --   & $14$ & $4$ \\
11 & $3$ & $9$  & $[27,7,15]^H_4$  & $6$ & $12$ & --   & $12$ & $3$ \\
12 & $3$ & $9$  & $[27,9,13]^H_4$  & $3$ & $9$  & --   & $9$  & $4$ \\
\hline
13 & $4$ & $5$  & $[20,4,13]^H_4$  & $8$ & $10$ & $12$ & $12$ & $1$ \\
\hline
\end{tabular}
\caption{Block-support bounds for balanced Hermitian LCD one-generator
quasi-cyclic codes over $\mathbb F_4$.  For these codes,
$\Psi(C)=+\infty$, $d_J(C)=\Phi_1(C)$, and
$d_{\mathrm{bsupp}}(C)=\Phi(C)=\Phi_{\ell-1}(C)$.}
\label{tab:balancedbounddemo}
\end{table}

\begin{table}[H]
\centering
\small
\renewcommand{\arraystretch}{1.15}
\setlength{\tabcolsep}{5pt}
\begin{tabular}{c c c c c c c c}
\hline
No. & $\ell$ & $m$ & Code $C$
& $\Phi(C)$
& $\Psi(C)$
& $d_{\mathrm{bsupp}}(C)$
& $d(C)-d_{\mathrm{bsupp}}(C)$ \\
\hline
1 & $2$ & $5$  & $[10,5,5]^H_4$   & $3$  & $5$  & $3$  & $2$ \\
2 & $2$ & $7$  & $[14,7,6]^H_4$   & $3$  & $7$  & $3$  & $3$ \\
3 & $2$ & $9$  & $[18,8,8]^H_4$   & $4$  & $9$  & $4$  & $4$ \\
4 & $2$ & $13$ & $[26,13,9]^H_4$  & $3$  & $13$ & $3$  & $6$ \\
5 & $2$ & $15$ & $[30,15,10]^H_4$ & $4$  & $10$ & $4$  & $6$ \\
\hline
6 & $4$ & $5$  & $[20,5,12]^H_4$  & $10$ & $15$ & $10$ & $2$ \\
\hline
7 & $6$ & $5$  & $[30,4,20]^H_4$  & $19$ & $20$ & $19$ & $1$ \\
\hline
\end{tabular}
\caption{Block-support bounds for unbalanced Hermitian LCD one-generator
quasi-cyclic codes over $\mathbb F_4$.}
\label{tab:unbalancedbounddemo}
\end{table}

\section{Quantum codes}

We now consider an application of Hermitian LCD codes in constructing entanglement-assisted quantum error-correcting codes (EAQECCs).
We note that EAQECCs generalize quantum stabilizer codes, as shown by Brun
et al.~\cite{brun2006}. 
We denote a $q$-ary EAQECC as $[[n,\kappa,\delta;c]]_q$, which
encodes $\kappa$ information qudits into $n$ channel qudits aided by $c$ pre-shared entangled pairs.

By \cite[Corollary 3.4]{guenda2018} and \cite{wildebrun2008}, 
every $[n,k,d]^{H}_{4}$ Hermitian LCD code yields a maximal-entanglement EAQECC with parameters $[[n,k,d;n-k]]_2$, so the codes obtained in Section~4 transfer directly to the quantum setting. 
 We compare the resulting codes with  the entanglement-assisted
table of \cite{grassl} and Table IV in \cite{LELLS26}, using   
Definition~\ref{partial-order} from \cite{LELLS26} below. 

\begin{definition}[{\cite[Sec.~I-B]{LELLS26}}]\label{partial-order}
To evaluate the performance of an $[[n,\kappa,\delta;c]]_q$ EAQECC, one
considers its rate, net rate, and error-correcting capacity, defined,
respectively, as
\[
\rho := \frac{\kappa}{n}, \qquad
\overline{\rho} := \frac{\kappa - c}{n}, \qquad
e := \left\lfloor \frac{\delta - 1}{2} \right\rfloor .
\]
EAQECCs with larger values of $\rho$, $\overline{\rho}$, and $e$ are
preferred. An EAQECC $Q_1$ is \emph{better} than $Q_2$, or $Q_1$ has
\emph{improved parameters} compared to $Q_2$, if $Q_1$ has at least an
improvement among the parameters $\rho$, $\overline{\rho}$, and $e$ when
the other parameters are fixed. Equivalently, an
$[[n_1,\kappa_1,\delta_1;c_1]]_q$ EAQECC $Q_1$ is better than an
$[[n_2,\kappa_2,\delta_2;c_2]]_q$ EAQECC $Q_2$ if
\[
n_1 \le n_2, \qquad \kappa_1 \ge \kappa_2, \qquad
\delta_1 \ge \delta_2, \qquad c_1 \le c_2,
\]
with at least one of these inequalities being strict.
\end{definition}

\subsection{Comparison with \cite{LELLS26}}

We compare the maximal-entanglement EAQECCs obtained from the
Hermitian LCD codes in Section~4 with the qubit EAQECCs recently reported in
\cite{LELLS26}. 
Under Definition~\ref{partial-order}, $13$ distinct codes obtained here strictly improve $24$ entries in \cite{LELLS26}. 
In several cases, a single code obtained here improves more than one entry in \cite{LELLS26}. 
The $24$ rows in Tables~\ref{tab:lells1} and~\ref{tab:lells2} count the improved entries of \cite{LELLS26}, rather than distinct codes constructed in this paper.

\begin{table}[H]
\centering 
\scriptsize
\renewcommand{\arraystretch}{1.12}
\resizebox{\textwidth}{!}{%
\begin{tabular}{c c c c c c c}
\hline
No. & Index \(\ell\) & \(m\) & Hermitian LCD code
& EAQECC obtained
& Code in~\cite{LELLS26}
& Improvement \\
\hline
1 & \(3\) & \(7\)
& \([21,7,11]^H_4\)
& \([[21,7,11;14]]_2\)
& \([[23,6,11;15]]_2\)
& \(n{-}2,\,\kappa{+}1,\,c{-}1\) \\
2 & \(3\) & \(7\)
& \([21,7,11]^H_4\)
& \([[21,7,11;14]]_2\)
& \([[23,7,10;14]]_2\)
& \(n{-}2,\,\delta{+}1\) \\
3 & \(3\) & \(7\)
& \([21,7,11]^H_4\)
& \([[21,7,11;14]]_2\)
& \([[26,7,11;15]]_2\)
& \(n{-}5,\,c{-}1\) \\
4 & \(2\) & \(13\)
& \([26,12,10]^H_4\)
& \([[26,12,10;14]]_2\)
& \([[26,8,10;14]]_2\)
& \(\kappa{+}4\) \\
5 & \(2\) & \(13\)
& \([26,12,10]^H_4\)
& \([[26,12,10;14]]_2\)
& \([[26,9,10;15]]_2\)
& \(\kappa{+}3,\,c{-}1\) \\
6 & \(2\) & \(13\)
& \([26,12,10]^H_4\)
& \([[26,12,10;14]]_2\)
& \([[27,9,10;14]]_2\)
& \(n{-}1,\,\kappa{+}3\) \\
7 & \(2\) & \(13\)
& \([26,12,10]^H_4\)
& \([[26,12,10;14]]_2\)
& \([[27,10,10;15]]_2\)
& \(n{-}1,\,\kappa{+}2,\,c{-}1\) \\
8 & \(2\) & \(13\)
& \([26,13,9]^H_4\)
& \([[26,13,9;13]]_2\)
& \([[27,12,9;13]]_2\)
& \(n{-}1,\,\kappa{+}1\) \\
9 & \(2\) & \(13\)
& \([26,12,10]^H_4\)
& \([[26,12,10;14]]_2\)
& \([[28,11,10;15]]_2\)
& \(n{-}2,\,\kappa{+}1,\,c{-}1\) \\
10 & \(3\) & \(9\)
& \([27,9,13]^H_4\)
& \([[27,9,13;18]]_2\)
& \([[29,9,12;18]]_2\)
& \(n{-}2,\,\delta{+}1\) \\
11 & \(2\) & \(13\)
& \([26,12,10]^H_4\)
& \([[26,12,10;14]]_2\)
& \([[30,11,10;15]]_2\)
& \(n{-}4,\,\kappa{+}1,\,c{-}1\) \\\hline
\end{tabular}%
}
\caption{Comparison with \cite{LELLS26} (lengths \(n \le 30\)).}
\label{tab:lells1}
\end{table}
\begin{table}[H]
\centering
\scriptsize
\renewcommand{\arraystretch}{1.12}
\resizebox{\textwidth}{!}{%
\begin{tabular}{c c c c c c c}
\hline
No. & Index \(\ell\) & \(m\) & Hermitian LCD code
& EAQECC obtained
& Code in~\cite{LELLS26}
& Improvement \\
\hline
12 & \(2\) & \(15\)
& \([30,15,10]^H_4\)
& \([[30,15,10;15]]_2\)
& \([[31,14,10;15]]_2\)
& \(n{-}1,\,\kappa{+}1\) \\
13 & \(3\) & \(11\)
& \([33,10,16]^H_4\)
& \([[33,10,16;23]]_2\)
& \([[33,9,15;24]]_2\)
& \(\kappa{+}1,\,\delta{+}1,\,c{-}1\) \\
14 & \(2\) & \(15\)
& \([30,15,10]^H_4\)
& \([[30,15,10;15]]_2\)
& \([[33,14,10;15]]_2\)
& \(n{-}3,\,\kappa{+}1\) \\
15 & \(2\) & \(15\)
& \([30,15,10]^H_4\)
& \([[30,15,10;15]]_2\)
& \([[33,15,10;16]]_2\)
& \(n{-}3,\,c{-}1\) \\
16 & \(3\) & \(11\)
& \([33,6,20]^H_4\)
& \([[33,6,20;27]]_2\)
& \([[34,6,19;28]]_2\)
& \(n{-}1,\,\delta{+}1,\,c{-}1\) \\
17 & \(3\) & \(11\)
& \([33,11,14]^H_4\)
& \([[33,11,14;22]]_2\)
& \([[34,10,14;22]]_2\)
& \(n{-}1,\,\kappa{+}1\) \\
18 & \(2\) & \(17\)
& \([34,8,18]^H_4\)
& \([[34,8,18;26]]_2\)
& \([[37,8,18;27]]_2\)
& \(n{-}3,\,c{-}1\) \\
19 & \(4\) & \(9\)
& \([36,9,18]^H_4\)
& \([[36,9,18;27]]_2\)
& \([[37,9,18;28]]_2\)
& \(n{-}1,\,c{-}1\) \\
20 & \(2\) & \(17\)
& \([34,17,11]^H_4\)
& \([[34,17,11;17]]_2\)
& \([[37,16,11;17]]_2\)
& \(n{-}3,\,\kappa{+}1\) \\
21 & \(4\) & \(9\)
& \([36,8,19]^H_4\)
& \([[36,8,19;28]]_2\)
& \([[38,8,19;30]]_2\)
& \(n{-}2,\,c{-}2\) \\
22 & \(4\) & \(9\)
& \([36,9,18]^H_4\)
& \([[36,9,18;27]]_2\)
& \([[38,9,18;29]]_2\)
& \(n{-}2,\,c{-}2\) \\
23 & \(2\) & \(19\)
& \([38,10,19]^H_4\)
& \([[38,10,19;28]]_2\)
& \([[38,10,17;28]]_2\)
& \(\delta{+}2\) \\
24 & \(2\) & \(17\)
& \([34,17,11]^H_4\)
& \([[34,17,11;17]]_2\)
& \([[39,17,11;18]]_2\)
& \(n{-}5,\,c{-}1\) \\\hline
\end{tabular}%
}
\caption{Comparison with \cite{LELLS26} (lengths \(n > 30\)).}
 \label{tab:lells2}
\end{table}

\subsection{Comparison  with Grassl's table \cite{grassl}}
 
The Hermitian LCD codes obtained in Section~4 yield $66$
maximal-entanglement EAQECCs with parameter sets not recorded in Grassl's table. 
We present these codes in Tables~\ref{tab:grasslindex2},
\ref{tab:grasslindex34}, and~\ref{tab:grasslindex567}, 
separating them according to their indices.
Among them, the code
\[
[[33,10,16;23]]_2
\]
improves the corresponding entry
\([[33,10,14\text{--}24;23]]_2\) in Grassl's table \cite{grassl}.
The remaining $65$ codes should be viewed primarily as new parameter sets, rather than as improvements over \cite{grassl}.

We compare each code with an entry in Grassl's table
having the same length $n$ and quantum dimension $\kappa$. 
For an entry
$ [[n,\kappa,d_{\mathrm L}\text{--}d_{\mathrm U};c_{\mathrm G}]]_2 $
in Grassl's table and an EAQECC $ [[n,\kappa,d;c]]_2 $
obtained from a Hermitian LCD code, we let
\[
\Delta d_{\mathrm L}:=d-d_{\mathrm L},
\qquad
\Delta c:=c-c_{\mathrm G}.
\]
Thus, \(\Delta d_{\mathrm L}\) records the increase over the lower
endpoint in Grassl's table, while \(\Delta c\) records the change in the number of required ebits. 
These two quantities should be read together.
In particular, a positive value of \(\Delta d_{\mathrm L}\) together with  a positive value of \(\Delta c\) represents a
distance-entanglement trade-off, rather than an improvement in the
sense of Definition~\ref{partial-order}.

\setcounter{grasslcode}{0}

\begin{table}[H]
\centering
\scriptsize
\renewcommand{\arraystretch}{1.12}
\resizebox{\textwidth}{!}{%
\begin{tabular}{c c c c c c c c}
\hline
No. & Index \(\ell\) & \(m\) & Hermitian LCD code
& EAQECC obtained
& Grassl table
& \(\Delta d_{\mathrm L}\)
& \(\Delta c\) \\
\hline
\nextgrasslno & \(2\) & \(13\)
& \([26,12,10]^H_4\)
& \([[26,12,10;14]]_2\)
& \([[26,12,7\text{--}11;7]]_2\)
& \(+3\) & \(+7\) \\

\nextgrasslno & \(2\) & \(13\)
& \([26,13,9]^H_4\)
& \([[26,13,9;13]]_2\)
& \([[26,13,7\text{--}11;7]]_2\)
& \(+2\) & \(+6\) \\

\nextgrasslno & \(2\) & \(15\)
& \([30,15,10]^H_4\)
& \([[30,15,10;15]]_2\)
& \([[30,15,8\text{--}15;13]]_2\)
& \(+2\) & \(+2\) \\

\nextgrasslno & \(2\) & \(17\)
& \([34,8,18]^H_4\)
& \([[34,8,18;26]]_2\)
& \([[34,8,15\text{--}26;23]]_2\)
& \(+3\) & \(+3\) \\

\nextgrasslno & \(2\) & \(17\)
& \([34,9,17]^H_4\)
& \([[34,9,17;25]]_2\)
& \([[34,9,15\text{--}25;23]]_2\)
& \(+2\) & \(+2\) \\

\nextgrasslno & \(2\) & \(17\)
& \([34,16,11]^H_4\)
& \([[34,16,11;18]]_2\)
& \([[34,16,9\text{--}16;13]]_2\)
& \(+2\) & \(+5\) \\

\nextgrasslno & \(2\) & \(17\)
& \([34,17,11]^H_4\)
& \([[34,17,11;17]]_2\)
& \([[34,17,8\text{--}16;13]]_2\)
& \(+3\) & \(+4\) \\

\nextgrasslno & \(2\) & \(19\)
& \([38,9,20]^H_4\)
& \([[38,9,20;29]]_2\)
& \([[38,9,17\text{--}28;23]]_2\)
& \(+3\) & \(+6\) \\

\nextgrasslno & \(2\) & \(19\)
& \([38,10,19]^H_4\)
& \([[38,10,19;28]]_2\)
& \([[38,10,16\text{--}27;23]]_2\)
& \(+3\) & \(+5\) \\

\nextgrasslno & \(2\) & \(19\)
& \([38,19,12]^H_4\)
& \([[38,19,12;19]]_2\)
& \([[38,19,9\text{--}19;17]]_2\)
& \(+3\) & \(+2\) \\

\nextgrasslno & \(2\) & \(21\)
& \([42,8,20]^H_4\)
& \([[42,8,20;34]]_2\)
& \([[42,8,17\text{--}33;23]]_2\)
& \(+3\) & \(+11\) \\

\nextgrasslno & \(2\) & \(21\)
& \([42,12,18]^H_4\)
& \([[42,12,18;30]]_2\)
& \([[42,12,17\text{--}30;26]]_2\)
& \(+1\) & \(+4\) \\

\nextgrasslno & \(2\) & \(23\)
& \([46,22,14]^H_4\)
& \([[46,22,14;24]]_2\)
& \([[46,22,13\text{--}23;21]]_2\)
& \(+1\) & \(+3\) \\

\nextgrasslno & \(2\) & \(25\)
& \([50,4,30]^H_4\)
& \([[50,4,30;46]]_2\)
& \([[50,4,23\text{--}43;4]]_2\)
& \(+7\) & \(+42\) \\

\nextgrasslno & \(2\) & \(25\)
& \([50,5,25]^H_4\)
& \([[50,5,25;45]]_2\)
& \([[50,5,20\text{--}41;5]]_2\)
& \(+5\) & \(+40\) \\

\nextgrasslno & \(2\) & \(25\)
& \([50,20,16]^H_4\)
& \([[50,20,16;30]]_2\)
& \([[50,20,14\text{--}26;21]]_2\)
& \(+2\) & \(+9\) \\

\nextgrasslno & \(2\) & \(25\)
& \([50,21,15]^H_4\)
& \([[50,21,15;29]]_2\)
& \([[50,21,14\text{--}26;21]]_2\)
& \(+1\) & \(+8\) \\

\nextgrasslno & \(2\) & \(25\)
& \([50,24,14]^H_4\)
& \([[50,24,14;26]]_2\)
& \([[50,24,12\text{--}24;20]]_2\)
& \(+2\) & \(+6\) \\

\nextgrasslno & \(2\) & \(25\)
& \([50,25,13]^H_4\)
& \([[50,25,13;25]]_2\)
& \([[50,25,11\text{--}23;19]]_2\)
& \(+2\) & \(+6\) \\
\hline
\end{tabular}%
}
\caption{Comparison with the Grassl EAQECC table for
index \(\ell=2\).}
\label{tab:grasslindex2}
\end{table}

\begin{table}[H]
\centering
\scriptsize
\renewcommand{\arraystretch}{1.12}
\resizebox{\textwidth}{!}{%
\begin{tabular}{c c c c c c c c}
\hline
No. & Index \(\ell\) & \(m\) & Hermitian LCD code
& EAQECC obtained
& Grassl table
& \(\Delta d_{\mathrm L}\)
& \(\Delta c\) \\
\hline
\nextgrasslno & \(3\) & \(7\)
& \([21,7,11]^H_4\)
& \([[21,7,11;14]]_2\)
& \([[21,7,9\text{--}14;11]]_2\)
& \(+2\) & \(+3\) \\

\nextgrasslno & \(3\) & \(9\)
& \([27,7,15]^H_4\)
& \([[27,7,15;20]]_2\)
& \([[27,7,9\text{--}18;11]]_2\)
& \(+6\) & \(+9\) \\

\nextgrasslno & \(3\) & \(9\)
& \([27,9,13]^H_4\)
& \([[27,9,13;18]]_2\)
& \([[27,9,9\text{--}18;16]]_2\)
& \(+4\) & \(+2\) \\

\nextgrasslno & \(3\) & \(11\)
& \([33,5,21]^H_4\)
& \([[33,5,21;28]]_2\)
& \([[33,5,14\text{--}28;23]]_2\)
& \(+7\) & \(+5\) \\

\nextgrasslno & \(3\) & \(11\)
& \([33,6,20]^H_4\)
& \([[33,6,20;27]]_2\)
& \([[33,6,14\text{--}27;23]]_2\)
& \(+6\) & \(+4\) \\

\nextgrasslno & \(3\) & \(11\)
& \([33,10,16]^H_4\)
& \([[33,10,16;23]]_2\)
& \([[33,10,14\text{--}24;23]]_2\)
& \(+2\) & \(0\) \\

\nextgrasslno & \(3\) & \(11\)
& \([33,11,14]^H_4\)
& \([[33,11,14;22]]_2\)
& \([[33,11,11\text{--}21;16]]_2\)
& \(+3\) & \(+6\) \\

\nextgrasslno & \(4\) & \(9\)
& \([36,4,23]^H_4\)
& \([[36,4,23;32]]_2\)
& \([[36,4,17\text{--}32;23]]_2\)
& \(+6\) & \(+9\) \\

\nextgrasslno & \(4\) & \(9\)
& \([36,5,22]^H_4\)
& \([[36,5,22;31]]_2\)
& \([[36,5,17\text{--}31;23]]_2\)
& \(+5\) & \(+8\) \\

\nextgrasslno & \(4\) & \(9\)
& \([36,6,21]^H_4\)
& \([[36,6,21;30]]_2\)
& \([[36,6,17\text{--}30;23]]_2\)
& \(+4\) & \(+7\) \\

\nextgrasslno & \(4\) & \(9\)
& \([36,7,19]^H_4\)
& \([[36,7,19;29]]_2\)
& \([[36,7,17\text{--}29;23]]_2\)
& \(+2\) & \(+6\) \\

\nextgrasslno & \(4\) & \(9\)
& \([36,8,19]^H_4\)
& \([[36,8,19;28]]_2\)
& \([[36,8,17\text{--}28;23]]_2\)
& \(+2\) & \(+5\) \\

\nextgrasslno & \(4\) & \(9\)
& \([36,9,18]^H_4\)
& \([[36,9,18;27]]_2\)
& \([[36,9,17\text{--}27;23]]_2\)
& \(+1\) & \(+4\) \\

\nextgrasslno & \(4\) & \(11\)
& \([44,5,29]^H_4\)
& \([[44,5,29;39]]_2\)
& \([[44,5,17\text{--}35;5]]_2\)
& \(+12\) & \(+34\) \\

\nextgrasslno & \(4\) & \(11\)
& \([44,6,28]^H_4\)
& \([[44,6,28;38]]_2\)
& \([[44,6,17\text{--}37;23]]_2\)
& \(+11\) & \(+15\) \\

\nextgrasslno & \(4\) & \(11\)
& \([44,10,21]^H_4\)
& \([[44,10,21;34]]_2\)
& \([[44,10,17\text{--}33;26]]_2\)
& \(+4\) & \(+8\) \\

\nextgrasslno & \(4\) & \(11\)
& \([44,11,20]^H_4\)
& \([[44,11,20;33]]_2\)
& \([[44,11,17\text{--}32;26]]_2\)
& \(+3\) & \(+7\) \\
\hline
\end{tabular}%
}
\caption{Comparison with the Grassl EAQECC table for
indices \(\ell=3,4\).}
\label{tab:grasslindex34}
\end{table}

\begin{table}[H]
\centering
\scriptsize
\renewcommand{\arraystretch}{1.12}
\resizebox{\textwidth}{!}{%
\begin{tabular}{c c c c c c c c}
\hline
No. & Index \(\ell\) & \(m\) & Hermitian LCD code
& EAQECC obtained
& Grassl table
& \(\Delta d_{\mathrm L}\)
& \(\Delta c\) \\
\hline
\nextgrasslno & \(5\) & \(7\)
& \([35,6,21]^H_4\)
& \([[35,6,21;29]]_2\)
& \([[35,6,16\text{--}29;23]]_2\)
& \(+5\) & \(+6\) \\

\nextgrasslno & \(5\) & \(7\)
& \([35,7,19]^H_4\)
& \([[35,7,19;28]]_2\)
& \([[35,7,16\text{--}28;23]]_2\)
& \(+3\) & \(+5\) \\

\nextgrasslno & \(5\) & \(9\)
& \([45,4,30]^H_4\)
& \([[45,4,30;41]]_2\)
& \([[45,4,20\text{--}38;4]]_2\)
& \(+10\) & \(+37\) \\

\nextgrasslno & \(5\) & \(9\)
& \([45,5,29]^H_4\)
& \([[45,5,29;40]]_2\)
& \([[45,5,18\text{--}36;5]]_2\)
& \(+11\) & \(+35\) \\

\nextgrasslno & \(5\) & \(9\)
& \([45,6,28]^H_4\)
& \([[45,6,28;39]]_2\)
& \([[45,6,17\text{--}38;23]]_2\)
& \(+11\) & \(+16\) \\

\nextgrasslno & \(5\) & \(9\)
& \([45,7,26]^H_4\)
& \([[45,7,26;38]]_2\)
& \([[45,7,17\text{--}37;23]]_2\)
& \(+9\) & \(+15\) \\

\nextgrasslno & \(5\) & \(9\)
& \([45,8,25]^H_4\)
& \([[45,8,25;37]]_2\)
& \([[45,8,17\text{--}36;23]]_2\)
& \(+8\) & \(+14\) \\

\nextgrasslno & \(5\) & \(9\)
& \([45,9,23]^H_4\)
& \([[45,9,23;36]]_2\)
& \([[45,9,17\text{--}34;23]]_2\)
& \(+6\) & \(+13\) \\

\nextgrasslno & \(5\) & \(11\)
& \([55,5,37]^H_4\)
& \([[55,5,37;50]]_2\)
& \([[55,5,23\text{--}46;5]]_2\)
& \(+14\) & \(+45\) \\

\nextgrasslno & \(5\) & \(11\)
& \([55,6,34]^H_4\)
& \([[55,6,34;49]]_2\)
& \([[55,6,21\text{--}44;6]]_2\)
& \(+13\) & \(+43\) \\

\nextgrasslno & \(5\) & \(11\)
& \([55,10,29]^H_4\)
& \([[55,10,29;45]]_2\)
& \([[55,10,17\text{--}43;26]]_2\)
& \(+12\) & \(+19\) \\

\nextgrasslno & \(5\) & \(11\)
& \([55,11,27]^H_4\)
& \([[55,11,27;44]]_2\)
& \([[55,11,17\text{--}41;26]]_2\)
& \(+10\) & \(+18\) \\

\nextgrasslno & \(6\) & \(7\)
& \([42,6,25]^H_4\)
& \([[42,6,25;36]]_2\)
& \([[42,6,17\text{--}35;23]]_2\)
& \(+8\) & \(+13\) \\

\nextgrasslno & \(6\) & \(7\)
& \([42,7,23]^H_4\)
& \([[42,7,23;35]]_2\)
& \([[42,7,17\text{--}34;23]]_2\)
& \(+6\) & \(+12\) \\

\nextgrasslno & \(6\) & \(9\)
& \([54,4,36]^H_4\)
& \([[54,4,36;50]]_2\)
& \([[54,4,24\text{--}47;4]]_2\)
& \(+12\) & \(+46\) \\

\nextgrasslno & \(6\) & \(9\)
& \([54,5,34]^H_4\)
& \([[54,5,34;49]]_2\)
& \([[54,5,22\text{--}45;5]]_2\)
& \(+12\) & \(+44\) \\

\nextgrasslno & \(6\) & \(9\)
& \([54,6,34]^H_4\)
& \([[54,6,34;48]]_2\)
& \([[54,6,20\text{--}43;6]]_2\)
& \(+14\) & \(+42\) \\

\nextgrasslno & \(6\) & \(9\)
& \([54,7,32]^H_4\)
& \([[54,7,32;47]]_2\)
& \([[54,7,20\text{--}41;7]]_2\)
& \(+12\) & \(+40\) \\

\nextgrasslno & \(6\) & \(9\)
& \([54,8,31]^H_4\)
& \([[54,8,31;46]]_2\)
& \([[54,8,18\text{--}39;8]]_2\)
& \(+13\) & \(+38\) \\

\nextgrasslno & \(6\) & \(9\)
& \([54,9,27]^H_4\)
& \([[54,9,27;45]]_2\)
& \([[54,9,17\text{--}42;23]]_2\)
& \(+10\) & \(+22\) \\

\nextgrasslno & \(7\) & \(5\)
& \([35,4,24]^H_4\)
& \([[35,4,24;31]]_2\)
& \([[35,4,16\text{--}31;23]]_2\)
& \(+8\) & \(+8\) \\

\nextgrasslno & \(7\) & \(5\)
& \([35,5,21]^H_4\)
& \([[35,5,21;30]]_2\)
& \([[35,5,16\text{--}30;23]]_2\)
& \(+5\) & \(+7\) \\

\nextgrasslno & \(7\) & \(7\)
& \([49,6,30]^H_4\)
& \([[49,6,30;43]]_2\)
& \([[49,6,18\text{--}38;6]]_2\)
& \(+12\) & \(+37\) \\

\nextgrasslno & \(7\) & \(7\)
& \([49,7,28]^H_4\)
& \([[49,7,28;42]]_2\)
& \([[49,7,17\text{--}40;20]]_2\)
& \(+11\) & \(+22\) \\

\nextgrasslno & \(7\) & \(9\)
& \([63,4,42]^H_4\)
& \([[63,4,42;59]]_2\)
& \([[63,4,28\text{--}56;4]]_2\)
& \(+14\) & \(+55\) \\

\nextgrasslno & \(7\) & \(9\)
& \([63,5,42]^H_4\)
& \([[63,5,42;58]]_2\)
& \([[63,5,27\text{--}54;5]]_2\)
& \(+15\) & \(+53\) \\

\nextgrasslno & \(7\) & \(9\)
& \([63,6,40]^H_4\)
& \([[63,6,40;57]]_2\)
& \([[63,6,24\text{--}52;6]]_2\)
& \(+16\) & \(+51\) \\

\nextgrasslno & \(7\) & \(9\)
& \([63,7,38]^H_4\)
& \([[63,7,38;56]]_2\)
& \([[63,7,24\text{--}50;7]]_2\)
& \(+14\) & \(+49\) \\

\nextgrasslno & \(7\) & \(9\)
& \([63,8,36]^H_4\)
& \([[63,8,36;55]]_2\)
& \([[63,8,23\text{--}48;8]]_2\)
& \(+13\) & \(+47\) \\

\nextgrasslno & \(7\) & \(9\)
& \([63,9,35]^H_4\)
& \([[63,9,35;54]]_2\)
& \([[63,9,21\text{--}46;9]]_2\)
& \(+14\) & \(+45\) \\
\hline
\end{tabular}%
}
\caption{Comparison with the Grassl EAQECC table for
indices \(\ell=5,6,7\).}
\label{tab:grasslindex567}
\end{table}

\section{Conclusion}

In this work, we have introduced block-support lower bounds for the
minimum distance of one-generator quasi-cyclic codes of arbitrary
index. These bounds are obtained from block-punctured and
block-shortened codes and apply to both balanced and unbalanced
one-generator quasi-cyclic codes.  In the balanced case, we have obtained a
nondecreasing chain of bounds whose first term recovers the Jensen bound.

We have also extended the hull-dimension results of
Aliabadi, G\"uneri and Kalayc{\i}~\cite{guneri2023} by giving a
polynomial description of the Hermitian hull itself and a necessary
and sufficient condition for a one-generator quasi-cyclic code to be
Hermitian LCD.
Using this criterion together with the block-support bounds as screening tools, we conducted a computer search for quaternary Hermitian LCD one-generator quasi-cyclic codes of lengths up to \(63\). 
We highlight the following outcomes:
\begin{itemize}
    \item six new Hermitian LCD codes
     \[
    [21,7,11]^H_4,\quad
    [26,12,10]^H_4,\quad
    [26,13,9]^H_4,\quad
    [27,7,15]^H_4,\quad
    [27,9,13]^H_4,\quad
    [30,15,10]^H_4,
    \]
    which improve the corresponding lower bounds in the Araya-Harada table. In addition, the code \([21,7,11]^H_4\) establishes
    \[
    d^H_4(21,7)=11;
    \]

    \item thirteen one-generator quasi-cyclic Hermitian LCD
    codes attaining exact values recorded in the Araya-Harada table;

    \item the optimal Hermitian LCD code \([35,4,24]^H_4\), which
    establishes
    \[
    d^H_4(35,4)=d_4(35,4)=24,
    \]
    together with ten further codes of length greater than \(30\)
    whose minimum distances equal those of the best-known
    unrestricted quaternary linear codes of the same length and
    dimension recorded in Grassl's table \cite{grassl}.
\end{itemize}

The Hermitian LCD codes obtained here also yield
maximal-entanglement EAQECCs. 
We have improved $24$ entries in \cite{LELLS26} with $13$ distinct EAQECCs constructed in this paper. 
We have also obtained $66$ maximal-entanglement EAQECCs with parameter sets not recorded in   Grassl's table \cite{grassl}. 
 In particular,
\[
[[33,10,16;23]]_2
\]
improves the corresponding entry in \cite{grassl}. 

A natural direction for future work is to extend the block-support
approach to multi-generator quasi-cyclic codes and to one-generator
quasi-twisted codes. 
It would also be worthwhile to compare the block-support bounds
systematically with other known minimum-distance bounds for
quasi-cyclic codes, particularly the Lally bound, the G\"uneri and \"Ozbudak bound, spectral bounds and
their recent refinements.

\section*{Acknowledgment} 

This work was supported by UAEU  grant G00004233.

\newpage

\appendix
\renewcommand{\thetable}{A\arabic{table}}
\setcounter{table}{0}
\section*{Appendix}

\setcounter{appendixcode}{0}

The tables in this appendix record the generator polynomials of the
quaternary Hermitian LCD one-generator quasi-cyclic codes obtained in our search, separating according to their indices. 
For a code of length $n=m\ell$ and index $\ell$, each row represents
\[
C=\left\langle
\bigl(g_1(x),g_2(x),\ldots,g_\ell(x)\bigr)
\right\rangle_R
\subseteq R^\ell,
\qquad
R=\mathbb F_4[x]/\langle x^m-1\rangle,
\]
and hence $C$ has length $n=m\ell$. We write
$\mathbb F_4=\{0,1,\omega,\omega^2\}$, where
$\omega^2+\omega+1=0$, and give the generator polynomials in
factorized form. 

The Remark column compares the minimum distance $d$ of $C$ with
that of an unrestricted quaternary linear code having the same length
and dimension. Here, $\mathrm{OpLC}$ denotes an optimal linear code,
so that
\[
d_{\mathrm{OpLC}}=d_4(n,k),
\]
whereas $\mathrm{BKLC}$ denotes a best-known linear code recorded in
Grassl's table~\cite{grassl}, whose optimality may not be known.

\begin{table}[H]
\centering
\fontsize{9.0}{9.8}\selectfont
\renewcommand{\arraystretch}{1.04}
\setlength{\tabcolsep}{1.2pt}

\begin{tabular}{
|>{\centering\arraybackslash}m{0.55cm}|
c|
>{\centering\arraybackslash}m{2.05cm}|
>{\raggedright\arraybackslash}m{4.65cm}|
>{\raggedright\arraybackslash}m{5.35cm}|
>{\centering\arraybackslash}m{1.75cm}|
}
\hline
\noalign{\hrule height 1.2pt}
No. & $m$ & $C$ & $g_{1}(x)$ & $g_{2}(x)$ & Remark \\
\noalign{\hrule height 1.2pt}

\nextcodeno
& 5
& $[10,4,6]^H_4$
& $x+1$
& $(x+1)^2(x+\omega)$
& $d=d_{\mathrm{OpLC}}$ \\
\cline{1-1}\cline{3-6}

\nextcodeno
&
& $[10,5,5]^H_4$
& $1$
& $(x+1)(x+\omega)^2$
& $d=d_{\mathrm{OpLC}}$ \\
\noalign{\hrule height 1.2pt}
\noalign{\hrule height 0.5pt}

\nextcodeno
& 7
& $[14,6,7]^H_4$
& $x+1$
& $(x+1)
  (x^4+\omega x^3+\omega^2x^2+\omega^2x+1)$
& $d=d_{\mathrm{OpLC}}$ \\
\cline{1-1}\cline{3-6}

\nextcodeno
&
& $[14,7,6]^H_4$
& $1$
& $(x+1)
  (x^3+\omega^2x^2+x+\omega)$
& $d=d_{\mathrm{OpLC}}$ \\
\noalign{\hrule height 1.2pt}
\noalign{\hrule height 0.5pt}

\nextcodeno
& 9
& $[18,6,9]^H_4$
& $(x+1)(x+\omega)(x+\omega^2)$
& $(x+1)(x+\omega)(x+\omega^2)
  (x^3+\omega x^2+\omega x+\omega)$
& $d=d_{\mathrm{OpLC}}-1$ \\
\cline{1-1}\cline{3-6}

\nextcodeno
&
& $[18,8,8]^H_4$
& $(x+1)(x+\omega)$
& $(x+1)(x+\omega^2)
  (x^3+x^2+x+\omega^2)$
& $d=d_{\mathrm{OpLC}}$ \\
\noalign{\hrule height 1.2pt}
\noalign{\hrule height 0.5pt}

\nextcodeno
& 11
& $[22,5,13]^H_4$
& $(x+1)
  (x^5+\omega x^4+x^3+x^2+\omega^2x+1)$
& $(x+1)(x+\omega)(x+\omega^2)
  (x^5+\omega x^4+x^3+x^2+\omega^2x+1)$
& $d=d_{\mathrm{OpLC}}-1$ \\
\cline{1-1}\cline{3-6}

\nextcodeno
&
& $[22,10,9]^H_4$
& $x+1$
& $(x+1)
  (x^6+\omega x^5+x^4+\omega x^3+
   \omega^2x^2+\omega^2x+1)$
& $d=d_{\mathrm{OpLC}}$ \\
\cline{1-1}\cline{3-6}

\nextcodeno
&
& $[22,11,8]^H_4$
& $1$
& $(x+1)^3(x^3+\omega^2x^2+1)$
& $d=d_{\mathrm{BKLC}}$ \\
\noalign{\hrule height 1.2pt}
\noalign{\hrule height 0.5pt}

\nextcodeno
& 13
& $[26,12,10]^H_4$
& $x+1$
& $(x+1)(x+\omega)(x+\omega^2)
  (x^2+\omega^2x+\omega^2)
  (x^3+\omega x+1)$
& $d=d_{\mathrm{OpLC}}-1$ \\
\cline{1-1}\cline{3-6}

\nextcodeno
&
& $[26,13,9]^H_4$
& $1$
& $(x+1)
  (x^7+\omega x^6+\omega^2x^5+\omega x^4+
   x^3+\omega x^2+\omega^2x+\omega^2)$
& $d=d_{\mathrm{OpLC}}-1$ \\
\noalign{\hrule height 1.2pt}
\noalign{\hrule height 0.5pt}

\nextcodeno
& 15
& $[30,8,15]^H_4$
& $(x+1)(x+\omega)(x+\omega^2)
  (x^2+x+\omega)(x^2+\omega x+\omega)$
& $(x+1)(x+\omega)(x+\omega^2)
  (x^2+x+\omega)(x^2+\omega x+\omega)
  (x^3+\omega^2x^2+\omega x+\omega)$
& $d=d_{\mathrm{BKLC}}-1$ \\
\cline{1-1}\cline{3-6}

\nextcodeno
&
& $[30,9,14]^H_4$
& $(x+1)(x+\omega)
  (x^2+x+\omega)(x^2+\omega x+\omega)$
& $(x+1)(x+\omega)^2(x+\omega^2)
  (x^2+x+\omega)(x^2+\omega x+\omega)
  (x^3+\omega)$
& $d=d_{\mathrm{BKLC}}-2$ \\
\cline{1-1}\cline{3-6}

\nextcodeno
&
& $[30,12,12]^H_4$
& $(x+1)(x+\omega)(x+\omega^2)$
& $(x+1)(x+\omega)(x+\omega^2)
  (x^7+\omega^2x^6+\omega^2x^5+
   \omega x^3+x^2+x+1)$
& $d=d_{\mathrm{BKLC}}$ \\
\cline{1-1}\cline{3-6}

\nextcodeno
&
& $[30,13,11]^H_4$
& $(x+1)(x+\omega)$
& $(x+1)(x+\omega)(x+\omega^2)
  (x^6+x^4+\omega^2x^3+\omega^2x^2+
   \omega x+\omega^2)$
& $d=d_{\mathrm{BKLC}}-1$ \\
\cline{1-1}\cline{3-6}

\nextcodeno
&
& $[30,14,10]^H_4$
& $x+1$
& $(x+1)^2(x+\omega)(x+\omega^2)
  (x^4+\omega x^3+x^2+\omega x+1)$
& $d=d_{\mathrm{OpLC}}-2$ \\
\cline{1-1}\cline{3-6}

\nextcodeno
&
& $[30,15,10]^H_4$
& $(x+1)(x+\omega)$
& $(x+\omega^2)
  (x^4+\omega^2x^3+\omega^2x^2+\omega x+\omega)
  (x^9+\omega^2x^7+x^4+x^3+
   \omega x^2+\omega^2x+1)$
& $d=d_{\mathrm{OpLC}}-2$ \\
\noalign{\hrule height 1.2pt}
\noalign{\hrule height 0.5pt}

\nextcodeno
& 17
& $[34,8,18]^H_4$
& $(x+1)
  (x^4+x^3+\omega x^2+x+1)
  (x^4+x^3+\omega^2x^2+x+1)$
& $(x+1)^3(x+\omega)(x+\omega^2)
  (x^2+\omega^2x+1)
  (x^4+x^3+\omega x^2+x+1)
  (x^4+x^3+\omega^2x^2+x+1)$
& $d=d_{\mathrm{BKLC}}-1$ \\
\cline{1-1}\cline{3-6}

\nextcodeno
&
& $[34,9,17]^H_4$
& $(x^4+x^3+\omega x^2+x+1)
  (x^4+x^3+\omega^2x^2+x+1)$
& $(x+1)^2(x+\omega)(x+\omega^2)
  (x^2+\omega^2x+1)
  (x^4+x^3+\omega x^2+x+1)
  (x^4+x^3+\omega^2x^2+x+1)$
& $d=d_{\mathrm{BKLC}}-1$ \\
\cline{1-1}\cline{3-6}

\nextcodeno
&
& $[34,16,11]^H_4$
& $x+1$
& $(x+1)^3
  (x^9+\omega^2x^8+\omega^2x^5+x^4+
   \omega^2x^3+\omega x^2+\omega^2)$
& $d=d_{\mathrm{BKLC}}-1$ \\
\cline{1-1}\cline{3-6}

\nextcodeno
&
& $[34,17,11]^H_4$
& $1$
& $(x+1)^2
  (x^9+\omega^2x^8+\omega^2x^5+x^4+
   \omega^2x^3+\omega x^2+\omega^2)$
& $d=d_{\mathrm{BKLC}}$ \\
\noalign{\hrule height 1.2pt}
\noalign{\hrule height 0.5pt}

\nextcodeno
& 19
& $[38,9,20]^H_4$
& $(x+1)
  (x^9+\omega x^8+\omega x^6+\omega x^5+
   \omega^2x^4+\omega^2x^3+\omega^2x+1)$
& $(x+1)^2(x+\omega^2)^2
  (x^9+\omega x^8+\omega x^6+\omega x^5+
   \omega^2x^4+\omega^2x^3+\omega^2x+1)$
& $d=d_{\mathrm{BKLC}}$ \\
\cline{1-1}\cline{3-6}

\nextcodeno
&
& $[38,10,19]^H_4$
& $x^9+\omega x^8+\omega x^6+\omega x^5+
   \omega^2x^4+\omega^2x^3+\omega^2x+1$
& $(x+1)
  (x^5+x^4+\omega^2x^3+x^2+x+1)
  (x^9+\omega x^8+\omega x^6+\omega x^5+
   \omega^2x^4+\omega^2x^3+\omega^2x+1)$
& $d=d_{\mathrm{BKLC}}$ \\
\cline{1-1}\cline{3-6}

\nextcodeno
&
& $[38,18,12]^H_4$
& $x+1$
& $(x+1)^2(x+\omega)(x+\omega^2)
  (x^4+\omega x^2+\omega x+\omega^2)
  (x^6+\omega^2x^5+x^4+\omega x^2+
   \omega x+\omega)$
& $d=d_{\mathrm{BKLC}}$ \\
\cline{1-1}\cline{3-6}

\nextcodeno
&
& $[38,19,12]^H_4$
& $1$
& $(x+1)(x+\omega)(x+\omega^2)
  (x^4+\omega x^2+\omega x+\omega^2)
  (x^6+\omega^2x^5+x^4+\omega x^2+
   \omega x+\omega)$
& $d=d_{\mathrm{BKLC}}$ \\

\noalign{\hrule height 1.2pt}
\noalign{\hrule height 0.5pt}
\end{tabular}

\caption{Hermitian LCD one-generator quasi-cyclic codes of index $2$.}
\label{tab:appendix-index2a}
\end{table}

\newpage

\begin{table}[H]
\centering
\fontsize{9.0}{9.8}\selectfont
\renewcommand{\arraystretch}{1.04}
\setlength{\tabcolsep}{1.2pt}

\begin{tabular}{
|>{\centering\arraybackslash}m{0.55cm}|
c|
>{\centering\arraybackslash}m{2.05cm}|
>{\raggedright\arraybackslash}m{4.65cm}|
>{\raggedright\arraybackslash}m{5.35cm}|
>{\centering\arraybackslash}m{1.75cm}|
}
\hline
\noalign{\hrule height 1.2pt}
No. & $m$ & $C$ & $g_{1}(x)$ & $g_{2}(x)$ & Remark \\
\noalign{\hrule height 1.2pt}

\nextcodeno
& 21
& $[42,6,21]^H_4$
& $(x+1)\,(x+\omega)\,(x+\omega^2)\,
  (x^3+x+1)\,(x^3+\omega^2x+1)\,
  (x^3+x^2+1)\,(x^3+\omega x^2+1)$
& $(x+1)\,(x+\omega)^3\,(x+\omega^2)\,
  (x^2+\omega x+\omega)\,
  (x^3+x+1)\,(x^3+\omega^2x+1)\,
  (x^3+x^2+1)\,(x^3+\omega x^2+1)$
& $d=d_{\mathrm{BKLC}}-7$ \\
\cline{1-1}\cline{3-6}

\nextcodeno
&
& $[42,7,20]^H_4$
& $(x+1)\,(x+\omega^2)\,
  (x^3+x+1)\,(x^3+\omega^2x+1)\,
  (x^3+x^2+1)\,(x^3+\omega x^2+1)$
& $(x+1)\,(x+\omega)^2\,(x+\omega^2)\,
  (x^2+\omega x+\omega)\,
  (x^3+x+1)\,(x^3+\omega^2x+1)\,
  (x^3+x^2+1)\,(x^3+\omega x^2+1)$
& $d=d_{\mathrm{BKLC}}-6$ \\
\cline{1-1}\cline{3-6}

\nextcodeno
&
& $[42,8,20]^H_4$
& $(x+1)\,(x+\omega^2)\,
  (x^3+x+1)\,(x^3+\omega^2x+1)\,
  (x^3+x^2+1)\,(x^3+\omega x^2+1)$
& $(x+\omega)\,(x+\omega^2)^2\,
  (x^2+x+\omega^2)\,
  (x^3+x+1)\,(x^3+\omega^2x+1)\,
  (x^3+x^2+1)\,(x^3+\omega x^2+1)$
& $d=d_{\mathrm{BKLC}}-5$ \\
\cline{1-1}\cline{3-6}

\nextcodeno
&
& $[42,12,18]^H_4$
& $(x+1)\,(x+\omega)\,(x+\omega^2)\,
  (x^3+x+1)\,(x^3+x^2+1)$
& $(x+1)^2\,(x+\omega)\,(x+\omega^2)^3\,
  (x^3+x+1)\,(x^3+x^2+1)\,
  (x^4+\omega^2x^2+\omega x+\omega)$
& $d=d_{\mathrm{BKLC}}-2$ \\
\cline{1-1}\cline{3-6}

\nextcodeno
&
& $[42,13,16]^H_4$
& $(x+\omega)\,(x+\omega^2)\,
  (x^3+x+1)\,(x^3+x^2+1)$
& $(x+1)\,(x+\omega)\,(x+\omega^2)^3\,
  (x^3+x+1)\,(x^3+x^2+1)\,
  (x^4+\omega^2x^2+\omega x+\omega)$
& $d=d_{\mathrm{BKLC}}-2$ \\
\cline{1-1}\cline{3-6}

\nextcodeno
&
& $[42,14,15]^H_4$
& $(x+\omega)\,(x+\omega^2)\,
  (x^3+x+1)\,(x^3+x^2+1)$
& $(x+1)\,(x+\omega^2)^2\,
  (x^3+x+1)\,(x^3+x^2+1)\,
  (x^4+\omega^2x^2+\omega x+\omega)$
& $d=d_{\mathrm{BKLC}}-3$ \\
\cline{1-1}\cline{3-6}

\nextcodeno
&
& $[42,15,14]^H_4$
& $(x+\omega)\,(x^3+x+1)\,(x^3+x^2+1)$
& $(x+1)\,(x+\omega^2)\,
  (x^3+x+1)\,(x^3+x^2+1)\,
  (x^4+\omega^2x^2+\omega x+\omega)$
& $d=d_{\mathrm{BKLC}}-3$ \\
\cline{1-1}\cline{3-6}

\nextcodeno
&
& $[42,18,14]^H_4$
& $(x+1)\,(x+\omega)\,(x+\omega^2)$
& $(x+1)\,(x+\omega)\,(x+\omega^2)\,
  (x^2+\omega^2x+\omega^2)\,
  (x^7+x^6+x^5+\omega x^4+\omega x^3+
   \omega x^2+x+1)$
& $d=d_{\mathrm{BKLC}}$ \\
\cline{1-1}\cline{3-6}

\nextcodeno
&
& $[42,19,14]^H_4$
& $(x+1)\,(x+\omega)$
& $(x+1)\,(x+\omega)\,(x+\omega^2)\,
  (x^{12}+\omega x^{11}+\omega^2x^{10}+
   \omega^2x^8+\omega x^7+\omega x^6+
   \omega^2x^5+\omega^2x^4+\omega x^2+x+\omega^2)$
& $d=d_{\mathrm{BKLC}}$ \\
\cline{1-1}\cline{3-6}

\nextcodeno
&
& $[42,20,12]^H_4$
& $x+1$
& $(x+1)^4\,(x+\omega)^3\,(x+\omega^2)\,
  (x^7+\omega x^5+x^4+\omega^2x^3+
   x^2+\omega^2x+1)$
& $d=d_{\mathrm{BKLC}}-2$ \\
\noalign{\hrule height 1.2pt}
\noalign{\hrule height 0.5pt}

\nextcodeno
& 23
& $[46,22,14]^H_4$
& $x+1$
& $(x+1)^3\,
  (x^{10}+\omega x^9+\omega x^8+\omega x^7+
   \omega^2x^6+x^5+\omega^2x^2+x+1)$
& $d=d_{\mathrm{BKLC}}$ \\
\cline{1-1}\cline{3-6}

\nextcodeno
&
& $[46,23,12]^H_4$
& $1$
& $(x+1)\,(x+\omega)^2\,
  (x^8+x^6+x^5+x^3+\omega^2x^2+\omega x+\omega)$
& $d=d_{\mathrm{BKLC}}-2$ \\

\noalign{\hrule height 1.2pt}
\noalign{\hrule height 0.5pt}

\nextcodeno
& 25
& $[50,4,30]^H_4$
& $(x+1)\,
  (x^{10}+\omega x^5+1)\,
  (x^{10}+\omega^2x^5+1)$
& $(x+1)\,
  (x^3+x^2+\omega^2x+\omega)\,
  (x^{10}+\omega x^5+1)\,
  (x^{10}+\omega^2x^5+1)$
& $d=d_{\mathrm{BKLC}}-6$ \\
\cline{1-1}\cline{3-6}

\nextcodeno
&
& $[50,5,25]^H_4$
& $(x^{10}+\omega x^5+1)\,
  (x^{10}+\omega^2x^5+1)$
& $(x+1)^2\,(x+\omega^2)\,
  (x^{10}+\omega x^5+1)\,
  (x^{10}+\omega^2x^5+1)$
& $d=d_{\mathrm{BKLC}}-10$ \\
\cline{1-1}\cline{3-6}

\nextcodeno
&
& $[50,20,16]^H_4$
& $(x+1)\,
  (x^2+\omega x+1)\,
  (x^2+\omega^2x+1)$
& $(x+1)\,
  (x^2+\omega x+1)\,
  (x^2+\omega^2x+1)\,
  (x^5+\omega x^4+\omega^2x^2+1)\,
  (x^{14}+x^{12}+x^{10}+\omega x^9+x^8+
   \omega^2x^6+\omega x^5+\omega^2x^4+
   x^3+\omega^2x^2+\omega^2)$
& $d=d_{\mathrm{BKLC}}-1$ \\
\cline{1-1}\cline{3-6}

\nextcodeno
&
& $[50,21,15]^H_4$
& $(x^2+\omega x+1)\,
  (x^2+\omega^2x+1)$
& $(x+1)^{10}\,(x+\omega^2)^2\,
  (x^2+\omega x+1)\,
  (x^2+\omega^2x+1)\,
  (x^3+\omega^2x^2+\omega^2x+\omega^2)$
& $d=d_{\mathrm{BKLC}}-1$ \\
\cline{1-1}\cline{3-6}

\nextcodeno
&
& $[50,24,14]^H_4$
& $(x+1)\,(x^2+\omega^2x+1)$
& $(x+1)^{20}\,(x+\omega)^8\,
  (x^3+\omega^2x^2+\omega^2x+\omega^2)$
& $d=d_{\mathrm{BKLC}}-1$ \\
\cline{1-1}\cline{3-6}

\nextcodeno
&
& $[50,25,13]^H_4$
& $x^2+\omega^2x+1$
& $(x+1)^{20}\,(x+\omega)^8\,
  (x^3+\omega^2x^2+\omega^2x+\omega^2)$
& $d=d_{\mathrm{BKLC}}-1$ \\

\noalign{\hrule height 1.2pt}
\noalign{\hrule height 0.5pt}
\end{tabular}

\caption{Hermitian LCD one-generator quasi-cyclic codes of index $2$
(continued).}
\label{tab:appendix-index2b}
\end{table}

\newpage


\begin{table}[H]
\centering
\fontsize{9.0}{9.8}\selectfont
\renewcommand{\arraystretch}{1.04}
\setlength{\tabcolsep}{1.0pt}

\begin{tabular}{
|>{\centering\arraybackslash}m{0.55cm}|
>{\centering\arraybackslash}m{0.45cm}|
>{\centering\arraybackslash}m{1.75cm}|
>{\raggedright\arraybackslash}m{2.35cm}|
>{\raggedright\arraybackslash}m{3.05cm}|
>{\raggedright\arraybackslash}m{4.50cm}|
>{\centering\arraybackslash}m{1.70cm}|
}
\hline
\noalign{\hrule height 1.2pt}
No. & $m$ & $C$ & $g_{1}(x)$ & $g_{2}(x)$ & $g_{3}(x)$ & Remark \\
\noalign{\hrule height 1.2pt}

\nextcodeno
& 5
& $[15,4,9]^H_4$
& $x+1$
& $(x+1)^2$
& $(x+1)\,(x+\omega^2)\,(x^2+x+\omega^2)$
& $d=d_{\mathrm{OpLC}}-1$ \\
\cline{1-1}\cline{3-7}

\nextcodeno
&
& $[15,5,8]^H_4$
& $1$
& $(x+\omega)\,(x+\omega^2)$
& $(x+\omega^2)\,(x^2+x+\omega)$
& $d=d_{\mathrm{OpLC}}$ \\
\noalign{\hrule height 1.2pt}
\noalign{\hrule height 0.5pt}

\nextcodeno
& 7
& $[21,6,12]^H_4$
& $x+1$
& $(x+1)\,(x+\omega)^3$
& $(x+1)\,(x+\omega)\,(x^3+x^2+x+\omega)$
& $d=d_{\mathrm{OpLC}}$ \\
\cline{1-1}\cline{3-7}

\nextcodeno
&
& $[21,7,11]^H_4$
& $1$
& $(x^2+x+\omega^2)\,(x^2+\omega^2x+1)$
& $(x^2+\omega x+\omega)\,(x^2+\omega^2x+1)$
& $d=d_{\mathrm{OpLC}}$ \\
\noalign{\hrule height 1.2pt}
\noalign{\hrule height 0.5pt}

\nextcodeno
& 9
& $[27,4,18]^H_4$
& $(x+1)\,(x+\omega)\,(x^3+\omega)$
& $(x+1)^2\,(x+\omega)\,(x^3+\omega)$
& $(x+1)\,(x+\omega)\,(x^2+\omega x+1)\,(x^3+\omega)$
& $d=d_{\mathrm{BKLC}}-1$ \\
\cline{1-1}\cline{3-7}

\nextcodeno
&
& $[27,6,15]^H_4$
& $(x+1)\,(x+\omega)\,(x+\omega^2)$
& $(x+1)\,(x+\omega)\,(x+\omega^2)\,
  (x^2+\omega x+\omega)$
& $(x+1)^3\,(x+\omega)^2\,(x+\omega^2)\,
  (x^2+x+\omega)$
& $d=d_{\mathrm{BKLC}}-1$ \\
\cline{1-1}\cline{3-7}

\nextcodeno
&
& $[27,7,15]^H_4$
& $(x+1)\,(x+\omega)$
& $(x+1)\,(x+\omega)\,(x^2+\omega x+\omega)$
& $(x+1)^3\,(x+\omega)^2\,(x^2+x+\omega)$
& $d=d_{\mathrm{BKLC}}$ \\
\cline{1-1}\cline{3-7}

\nextcodeno
&
& $[27,8,13]^H_4$
& $x+\omega^2$
& $(x+\omega^2)\,
  (x^4+\omega^2x^3+\omega^2x+\omega^2)$
& $(x+\omega^2)\,
  (x^7+\omega x^6+\omega^2x^5+x^4+x^3+x^2+
   \omega x+\omega)$
& $d=d_{\mathrm{BKLC}}-1$ \\
\cline{1-1}\cline{3-7}

\nextcodeno
&
& $[27,9,13]^H_4$
& $1$
& $x^4+\omega^2x^3+\omega^2x+\omega^2$
& $x^7+\omega x^6+\omega^2x^5+x^4+x^3+x^2+
   \omega x+\omega$
& $d=d_{\mathrm{BKLC}}$ \\
\noalign{\hrule height 1.2pt}
\noalign{\hrule height 0.5pt}

\nextcodeno
& 11
& $[33,5,21]^H_4$
& $(x+1)\,
  (x^5+\omega^2x^4+x^3+x^2+\omega x+1)$
& $(x+1)^2\,(x+\omega^2)^2\,
  (x^5+\omega^2x^4+x^3+x^2+\omega x+1)$
& $(x+1)\,(x+\omega)\,
  (x^5+\omega^2x^4+x^3+x^2+\omega x+1)$
& $d=d_{\mathrm{BKLC}}-1$ \\
\cline{1-1}\cline{3-7}

\nextcodeno
&
& $[33,6,20]^H_4$
& $x^5+\omega x^4+x^3+x^2+\omega^2x+1$
& $(x+\omega)\,(x+\omega^2)\,
  (x^5+\omega x^4+x^3+x^2+\omega^2x+1)$
& $(x+\omega^2)\,
  (x^3+x^2+\omega^2x+\omega)\,
  (x^5+\omega x^4+x^3+x^2+\omega^2x+1)$
& $d=d_{\mathrm{BKLC}}$ \\
\cline{1-1}\cline{3-7}

\nextcodeno
&
& $[33,10,16]^H_4$
& $x+1$
& $(x+1)^2\,(x+\omega)\,(x^2+x+\omega^2)$
& $(x+1)^2\,
  (x^2+\omega x+\omega)\,
  (x^5+x^3+\omega^2x^2+\omega x+\omega)$
& $d=d_{\mathrm{BKLC}}$ \\
\cline{1-1}\cline{3-7}

\nextcodeno
&
& $[33,11,14]^H_4$
& $1$
& $(x+\omega)\,(x^2+x+\omega)$
& $(x+\omega^2)^2\,
  (x^7+x^6+x^3+\omega^2x^2+\omega^2x+\omega^2)$
& $d=d_{\mathrm{BKLC}}-1$ \\

\noalign{\hrule height 1.2pt}
\noalign{\hrule height 0.5pt}
\end{tabular}

\caption{Hermitian LCD one-generator quasi-cyclic codes of index $3$.}
\label{tab:appendix-index3}
\end{table}

\newpage


\begin{table}[H]
\centering
\fontsize{9.0}{9.6}\selectfont
\renewcommand{\arraystretch}{1.02}
\setlength{\tabcolsep}{0.5pt}

\begin{tabular}{
|>{\centering\arraybackslash}m{0.48cm}|
>{\centering\arraybackslash}m{0.42cm}|
>{\centering\arraybackslash}m{1.85cm}|
>{\raggedright\arraybackslash}m{2.40cm}|
>{\raggedright\arraybackslash}m{2.40cm}|
>{\raggedright\arraybackslash}m{2.70cm}|
>{\raggedright\arraybackslash}m{3.15cm}|
>{\centering\arraybackslash}m{1.45cm}|
}
\hline
\noalign{\hrule height 1.2pt}
No. & $m$ & $C$ & $g_{1}(x)$ & $g_{2}(x)$
& $g_{3}(x)$ & $g_{4}(x)$ & Remark \\
\noalign{\hrule height 1.2pt}

\nextcodeno
& 5
& $[20,4,13]^H_4$
& $x+1$
& $(x+1)^2$
& $(x+1)\,\allowbreak(x+\omega)^2$
& $(x+1)^2\,\allowbreak
  (x^2+\omega^2x+\omega^2)$
& $d=d_{\mathrm{OpLC}}$ \\
\cline{1-1}\cline{3-8}

\nextcodeno
&
& $[20,5,12]^H_4$
& $1$
& $(x^2+\omega^2x+\omega^2)$
& $(x+1)^3$
& $(x^3+\omega^2x^2+x+\omega)$
& $d=d_{\mathrm{OpLC}}$ \\
\noalign{\hrule height 1.2pt}
\noalign{\hrule height 0.5pt}

\nextcodeno
& 7
& $[28,7,15]^H_4$
& $1$
& $(x+1)\,\allowbreak
  (x^2+\omega^2x+1)$
& $(x^3+x+1)$
& $(x^6+\omega^2x^5+\omega x^4+\omega^2x^3+
   x^2+\omega x+\omega)$
& $d=d_{\mathrm{OpLC}}-1$ \\
\noalign{\hrule height 1.2pt}
\noalign{\hrule height 0.5pt}

\nextcodeno
& 9
& $[36,4,23]^H_4$
& $(x+1)\,\allowbreak
  (x+\omega)\,\allowbreak
  (x+\omega^2)\,\allowbreak
  (x^3+\omega^2)$
& $(x+1)\,\allowbreak
  (x+\omega)\,\allowbreak
  (x^3+\omega^2)$
& $(x+1)^2\,\allowbreak
  (x+\omega)\,\allowbreak
  (x^3+\omega^2)$
& $(x+1)\,\allowbreak
  (x+\omega)\,\allowbreak
  (x^2+\omega x+1)\,\allowbreak
  (x^3+\omega^2)$
& $d=d_{\mathrm{BKLC}}-2$ \\
\cline{1-1}\cline{3-8}

\nextcodeno
&
& $[36,5,22]^H_4$
& $(x+\omega)\,\allowbreak
  (x+1)\,\allowbreak
  (x^3+\omega^2)$
& $(x+\omega)\,\allowbreak
  (x+\omega^2)\,\allowbreak
  (x^3+\omega^2)$
& $(x+\omega)\,\allowbreak
  (x^2+\omega^2x+\omega^2)\,\allowbreak
  (x^3+\omega^2)$
& $(x+\omega)^4\,\allowbreak
  (x^3+\omega^2)$
& $d=d_{\mathrm{BKLC}}-2$ \\
\cline{1-1}\cline{3-8}

\nextcodeno
&
& $[36,6,21]^H_4$
& $(x+1)\,\allowbreak
  (x^3+\omega^2)$
& $(x+\omega^2)\,\allowbreak
  (x^3+\omega^2)$
& $(x^2+\omega^2x+\omega^2)\,\allowbreak
  (x^3+\omega^2)$
& $(x+\omega)^3\,\allowbreak
  (x^3+\omega^2)$
& $d=d_{\mathrm{BKLC}}-2$ \\
\cline{1-1}\cline{3-8}

\nextcodeno
&
& $[36,7,19]^H_4$
& $(x+1)\,\allowbreak
  (x+\omega)\,\allowbreak
  (x+\omega^2)$
& $(x+1)\,\allowbreak
  (x+\omega^2)$
& $(x+1)^2\,\allowbreak
  (x+\omega^2)^2\,\allowbreak
  (x^3+\omega^2x^2+\omega x+\omega^2)$
& $(x+1)^2\,\allowbreak
  (x+\omega^2)\,\allowbreak
  (x^2+\omega x+1)\,\allowbreak
  (x^3+x^2+1)$
& $d=d_{\mathrm{BKLC}}-3$ \\
\cline{1-1}\cline{3-8}

\nextcodeno
&
& $[36,8,19]^H_4$
& $(x+1)\,\allowbreak
  (x+\omega)$
& $x+1$
& $(x+1)^2\,\allowbreak
  (x+\omega^2)\,\allowbreak
  (x^3+\omega^2x^2+\omega x+\omega^2)$
& $(x+1)^2\,\allowbreak
  (x^2+\omega x+1)\,\allowbreak
  (x^3+x^2+1)$
& $d=d_{\mathrm{BKLC}}-1$ \\
\cline{1-1}\cline{3-8}

\nextcodeno
&
& $[36,9,18]^H_4$
& $1$
& $(x+\omega^2)\,\allowbreak
  (x^3+\omega)$
& $(x^4+\omega^2x^3+x^2+1)$
& $(x+1)\,\allowbreak
  (x+\omega)^2\,\allowbreak
  (x^5+\omega x^4+\omega^2x^2+1)$
& $d=d_{\mathrm{BKLC}}-1$ \\
\noalign{\hrule height 1.2pt}
\noalign{\hrule height 0.5pt}

\nextcodeno
& 11
& $[44,5,29]^H_4$
& $(x+1)\,\allowbreak
  (x^5+\omega^2x^4+x^3+x^2+\omega x+1)$
& $(x+1)\,\allowbreak
  (x^5+\omega^2x^4+x^3+x^2+\omega x+1)$
& $(x+1)\,\allowbreak
  (x+\omega)\,\allowbreak
  (x+\omega^2)\,\allowbreak
  (x^5+\omega^2x^4+x^3+x^2+\omega x+1)$
& $(x+1)^2\,\allowbreak
  (x+\omega^2)^2\,\allowbreak
  (x^5+\omega^2x^4+x^3+x^2+\omega x+1)$
& $d=d_{\mathrm{BKLC}}-1$ \\
\cline{1-1}\cline{3-8}

\nextcodeno
&
& $[44,6,28]^H_4$
& $(x^5+\omega^2x^4+x^3+x^2+\omega x+1)$
& $(x+1)\,\allowbreak
  (x^5+\omega^2x^4+x^3+x^2+\omega x+1)$
& $(x+\omega^2)\,\allowbreak
  (x^2+x+\omega)\,\allowbreak
  (x^5+\omega^2x^4+x^3+x^2+\omega x+1)$
& $(x^4+\omega^2x^3+\omega x^2+\omega)\,
  \allowbreak
  (x^5+\omega^2x^4+x^3+x^2+\omega x+1)$
& $d=d_{\mathrm{BKLC}}-1$ \\
\cline{1-1}\cline{3-8}

\nextcodeno
&
& $[44,10,21]^H_4$
& $x+1$
& $(x+1)\,\allowbreak
  (x+\omega)^2$
& $(x+1)\,\allowbreak
  (x^5+\omega x^4+x^3+x^2+\omega^2x+1)$
& $(x+1)\,\allowbreak
  (x+\omega)\,\allowbreak
  (x+\omega^2)\,\allowbreak
  (x^2+\omega^2x+1)\,\allowbreak
  (x^4+x^3+x^2+\omega)$
& $d=d_{\mathrm{BKLC}}-2$ \\
\cline{1-1}\cline{3-8}

\nextcodeno
&
& $[44,11,20]^H_4$
& $1$
& $(x+\omega)^7\,\allowbreak
  (x+\omega^2)^3$
& $(x^5+\omega x^4+x^3+x^2+\omega^2x+1)$
& $(x+1)^2\,\allowbreak
  (x+\omega^2)\,\allowbreak
  (x^2+\omega^2x+1)\,\allowbreak
  (x^4+x^3+x^2+\omega)$
& $d=d_{\mathrm{BKLC}}-2$ \\

\noalign{\hrule height 1.2pt}
\noalign{\hrule height 0.5pt}
\end{tabular}

\caption{Hermitian LCD one-generator quasi-cyclic codes of index $4$.}
\label{tab:appendix-index4}
\end{table}

\newpage


\clearpage
\begin{landscape}

\begin{table}[H]
\centering
\fontsize{9.0}{9.3}\selectfont
\renewcommand{\arraystretch}{1.0}
\setlength{\tabcolsep}{0.7pt}

\begin{tabular}{
|>{\centering\arraybackslash}m{0.55cm}|
>{\centering\arraybackslash}m{0.42cm}|
>{\centering\arraybackslash}m{1.75cm}|
>{\raggedright\arraybackslash}m{3.40cm}|
>{\raggedright\arraybackslash}m{3.40cm}|
>{\raggedright\arraybackslash}m{3.40cm}|
>{\raggedright\arraybackslash}m{3.40cm}|
>{\raggedright\arraybackslash}m{3.40cm}|
>{\centering\arraybackslash}m{1.55cm}|
}
\hline
\noalign{\hrule height 1.2pt}
No. & $m$ & $C$ & $g_{1}(x)$ & $g_{2}(x)$ & $g_{3}(x)$
& $g_{4}(x)$ & $g_{5}(x)$ & Remark \\
\noalign{\hrule height 1.2pt}

\nextcodeno
& 5
& $[25,5,15]^H_4$
& $1$
& $x+1$
& $x^2+\omega^2x+\omega^2$
& $x^3+\omega x^2+\omega^2x+\omega^2$
& $(x+1)^2\,\allowbreak
  (x^2+\omega^2x+\omega^2)$
& $d=d_{\mathrm{OpLC}}-1$ \\
\noalign{\hrule height 1.2pt}
\noalign{\hrule height 0.5pt}

\nextcodeno
& 7
& $[35,6,21]^H_4$
& $x+1$
& $(x+1)\,\allowbreak(x+\omega)^2$
& $(x+1)\,\allowbreak(x+\omega)\,\allowbreak
  (x^2+x+\omega^2)$
& $(x+1)\,\allowbreak(x+\omega)\,\allowbreak
  (x+\omega^2)\,\allowbreak
  (x^2+x+\omega)$
& $(x+1)^2\,\allowbreak
  (x+\omega^2)\,\allowbreak
  (x^2+\omega x+\omega)$
& $d=d_{\mathrm{BKLC}}-1$ \\
\cline{1-1}\cline{3-9}

\nextcodeno
&
& $[35,7,19]^H_4$
& $1$
& $(x+\omega)^2\,\allowbreak(x+\omega^2)$
& $(x+\omega)^4\,\allowbreak(x+\omega^2)^2$
& $(x+\omega^2)\,\allowbreak
  (x^5+\omega x^3+\omega^2x+1)$
& $(x^2+x+\omega^2)\,\allowbreak
  (x^2+\omega^2x+1)$
& $d=d_{\mathrm{BKLC}}-2$ \\
\noalign{\hrule height 1.2pt}
\noalign{\hrule height 0.5pt}

\nextcodeno
& 9
& $[45,4,30]^H_4$
& $(x+1)\,\allowbreak
  (x+\omega^2)\,\allowbreak
  (x^3+\omega)$
& $(x+1)\,\allowbreak
  (x+\omega^2)^3\,\allowbreak
  (x^3+\omega)$
& $(x+1)\,\allowbreak
  (x+\omega^2)\,\allowbreak
  (x^2+\omega x+\omega)\,\allowbreak
  (x^3+\omega)$
& $(x+1)\,\allowbreak
  (x+\omega^2)^2\,\allowbreak
  (x^3+\omega)$
& $(x+1)\,\allowbreak
  (x+\omega^2)\,\allowbreak
  (x^3+\omega)\,\allowbreak
  (x^3+\omega x+1)$
& $d=d_{\mathrm{OpLC}}-2$ \\
\cline{1-1}\cline{3-9}

\nextcodeno
&
& $[45,5,29]^H_4$
& $(x+\omega^2)\,\allowbreak
  (x^3+\omega)$
& $(x+\omega^2)\,\allowbreak
  (x^2+x+\omega^2)\,\allowbreak
  (x^2+\omega^2x+1)\,\allowbreak
  (x^3+\omega)$
& $(x+\omega^2)\,\allowbreak
  (x^2+\omega x+\omega)\,\allowbreak
  (x^3+\omega)$
& $(x+\omega^2)^2\,\allowbreak
  (x^3+\omega)$
& $(x+\omega^2)^5\,\allowbreak
  (x^3+\omega)$
& $d=d_{\mathrm{OpLC}}-2$ \\
\cline{1-1}\cline{3-9}

\nextcodeno
&
& $[45,6,28]^H_4$
& $(x+1)\,\allowbreak
  (x+\omega)\,\allowbreak
  (x+\omega^2)$
& $(x+1)^5\,\allowbreak
  (x+\omega)\,\allowbreak
  (x+\omega^2)^2$
& $(x+1)\,\allowbreak
  (x+\omega)^2\,\allowbreak
  (x+\omega^2)\,\allowbreak
  (x^4+x^2+\omega x+1)$
& $(x+1)\,\allowbreak
  (x+\omega)\,\allowbreak
  (x+\omega^2)^2\,\allowbreak
  (x^4+x^2+\omega^2x+1)$
& $(x+1)^2\,\allowbreak
  (x+\omega)\,\allowbreak
  (x+\omega^2)^3$
& $d=d_{\mathrm{OpLC}}-2$ \\
\cline{1-1}\cline{3-9}

\nextcodeno
&
& $[45,7,26]^H_4$
& $(x+\omega)\,\allowbreak
  (x+\omega^2)$
& $(x+\omega)\,\allowbreak
  (x+\omega^2)\,\allowbreak
  (x^6+\omega x^5+\omega^2x^3+x^2+\omega x+1)$
& $(x+\omega)\,\allowbreak
  (x+\omega^2)\,\allowbreak
  (x^2+x+\omega^2)\,\allowbreak
  (x^4+x^3+x^2+\omega^2x+\omega^2)$
& $(x+\omega)\,\allowbreak
  (x+\omega^2)\,\allowbreak
  (x^5+\omega^2x^4+\omega x^3+\omega^2x^2+x+1)$
& $(x+\omega)^6\,\allowbreak
  (x+\omega^2)^2$
& $d=d_{\mathrm{BKLC}}-2$ \\
\cline{1-1}\cline{3-9}

\nextcodeno
&
& $[45,8,25]^H_4$
& $(x+\omega)\,\allowbreak
  (x+\omega^2)$
& $(x+1)^6\,\allowbreak
  (x+\omega^2)$
& $(x+1)\,\allowbreak
  (x+\omega^2)\,\allowbreak
  (x^3+x^2+x+\omega^2)$
& $(x+\omega)\,\allowbreak
  (x+\omega^2)\,\allowbreak
  (x^6+\omega x^5+\omega x^4+\omega x^3+
   \omega^2x^2+\omega)$
& $(x+\omega^2)\,\allowbreak
  (x^6+\omega x^5+x^4+\omega^2x^3+\omega x+1)$
& $d=d_{\mathrm{BKLC}}-2$ \\
\cline{1-1}\cline{3-9}

\nextcodeno
&
& $[45,9,23]^H_4$
& $x+\omega$
& $(x+1)^6$
& $(x+1)\,\allowbreak
  (x^7+\omega^2x^6+x^5+x^3+x+\omega)$
& $(x+\omega)\,\allowbreak
  (x^7+\omega^2x^6+\omega^2x^5+\omega x^4+
   \omega x^3+x^2+x+\omega)$
& $x^6+\omega x^5+x^4+\omega^2x^3+\omega x+1$
& $d=d_{\mathrm{BKLC}}-2$ \\
\noalign{\hrule height 1.2pt}
\noalign{\hrule height 0.5pt}

\nextcodeno
& 11
& $[55,5,37]^H_4$
& $(x+1)\,\allowbreak
  (x^5+\omega^2x^4+x^3+x^2+\omega x+1)$
& $(x+1)\,\allowbreak
  (x^2+\omega x+\omega)\,\allowbreak
  (x^2+\omega^2x+\omega^2)\,\allowbreak
  (x^5+\omega^2x^4+x^3+x^2+\omega x+1)$
& $(x+1)\,\allowbreak
  (x+\omega)\,\allowbreak
  (x^3+x^2+x+\omega^2)\,\allowbreak
  (x^5+\omega^2x^4+x^3+x^2+\omega x+1)$
& $(x+1)\,\allowbreak
  (x^2+x+\omega)\,\allowbreak
  (x^5+\omega^2x^4+x^3+x^2+\omega x+1)$
& $(x+1)^2\,\allowbreak
  (x^3+\omega x+1)\,\allowbreak
  (x^5+\omega^2x^4+x^3+x^2+\omega x+1)$
& $d=d_{\mathrm{OpLC}}-2$ \\
\cline{1-1}\cline{3-9}

\nextcodeno
&
& $[55,6,34]^H_4$
& $x^5+\omega^2x^4+x^3+x^2+\omega x+1$
& $(x+\omega)^2\,\allowbreak
  (x^5+\omega^2x^4+x^3+x^2+\omega x+1)$
& $(x^3+\omega^2x^2+1)\,\allowbreak
  (x^5+\omega^2x^4+x^3+x^2+\omega x+1)$
& $(x+\omega)\,\allowbreak
  (x+\omega^2)\,\allowbreak
  (x^2+x+\omega)\,\allowbreak
  (x^5+\omega^2x^4+x^3+x^2+\omega x+1)$
& $(x^2+\omega^2x+1)\,\allowbreak
  (x^5+\omega^2x^4+x^3+x^2+\omega x+1)$
& $d=d_{\mathrm{BKLC}}-2$ \\
\cline{1-1}\cline{3-9}

\nextcodeno
&
& $[55,10,29]^H_4$
& $(x+1)\,\allowbreak
  (x^5+\omega^2x^4+x^3+x^2+\omega x+1)$
& $(x+1)\,\allowbreak
  (x+\omega)^2$
& $(x+1)\,\allowbreak
  (x+\omega)\,\allowbreak
  (x^7+\omega^2x^6+x^5+x^3+x+\omega)$
& $(x+1)^2\,\allowbreak
  (x+\omega)\,\allowbreak
  (x^2+x+\omega)$
& $(x+1)\,\allowbreak
  (x+\omega)\,\allowbreak
  (x^2+\omega^2x+1)\,\allowbreak
  (x^5+x^4+\omega^2)$
& $d=d_{\mathrm{BKLC}}-2$ \\
\cline{1-1}\cline{3-9}

\nextcodeno
&
& $[55,11,27]^H_4$
& $x^5+\omega^2x^4+x^3+x^2+\omega x+1$
& $(x+\omega)^2$
& $(x+\omega)\,\allowbreak
  (x^7+\omega^2x^6+x^5+x^3+x+\omega)$
& $(x+1)\,\allowbreak
  (x+\omega)\,\allowbreak
  (x^2+x+\omega)$
& $(x+1)\,\allowbreak
  (x+\omega^2)\,\allowbreak
  (x^2+\omega^2x+\omega^2)\,\allowbreak
  (x^3+\omega^2)$
& $d=d_{\mathrm{BKLC}}-2$ \\

\noalign{\hrule height 1.2pt}
\noalign{\hrule height 0.5pt}
\end{tabular}

\caption{Hermitian LCD one-generator quasi-cyclic codes of index $5$.}
\label{tab:appendix-index5}
\end{table}

\end{landscape}
\clearpage


\clearpage
\begin{landscape}

\begin{table}[H]
\centering
\fontsize{9.0}{9.4}\selectfont
\renewcommand{\arraystretch}{1.02}
\setlength{\tabcolsep}{0.45pt}

\begin{tabular}{
|>{\centering\arraybackslash}m{0.52cm}|
>{\centering\arraybackslash}m{0.36cm}|
>{\centering\arraybackslash}m{2.00cm}|
>{\raggedright\arraybackslash}m{2.70cm}|
>{\raggedright\arraybackslash}m{2.70cm}|
>{\raggedright\arraybackslash}m{2.70cm}|
>{\raggedright\arraybackslash}m{2.70cm}|
>{\raggedright\arraybackslash}m{2.70cm}|
>{\raggedright\arraybackslash}m{2.70cm}|
>{\centering\arraybackslash}m{1.42cm}|
}
\hline
\noalign{\hrule height 1.2pt}
No. & $m$ & $C$ & $g_{1}(x)$ & $g_{2}(x)$ & $g_{3}(x)$
& $g_{4}(x)$ & $g_{5}(x)$ & $g_{6}(x)$ & Remark \\
\noalign{\hrule height 1.2pt}

\nextcodeno
& 5
& $[30,4,20]^H_4$
& $x+1$
& $(x+1)\,\allowbreak(x+\omega)^2\,\allowbreak(x+\omega^2)$
& $(x+1)^2\,\allowbreak(x+\omega)^2$
& $(x+1)^2$
& $(x+1)\,\allowbreak(x+\omega)\,\allowbreak
  (x^2+\omega^2x+1)$
& $(x+1)\,\allowbreak(x+\omega)\,\allowbreak
  (x^2+x+\omega)$
& $d=d_{\mathrm{OpLC}}-1$ \\
\cline{1-1}\cline{3-10}

\nextcodeno
&
& $[30,5,19]^H_4$
& $1$
& $(x+1)\,\allowbreak
  (x^3+\omega^2x^2+\omega x+\omega^2)$
& $x^3+x^2+1$
& $(x+\omega)\,\allowbreak(x^3+\omega^2)$
& $x^3+\omega^2x^2+\omega x+\omega$
& $x+\omega$
& $d=d_{\mathrm{OpLC}}-1$ \\
\noalign{\hrule height 1.2pt}
\noalign{\hrule height 0.5pt}

\nextcodeno
& 7
& $[42,6,25]^H_4$
& $x+1$
& $(x+1)\,\allowbreak(x+\omega)^2$
& $(x+1)\,\allowbreak(x+\omega)\,\allowbreak
  (x^2+x+\omega^2)$
& $(x+1)\,\allowbreak(x+\omega)\,\allowbreak
  (x+\omega^2)\,\allowbreak(x^2+x+\omega)$
& $(x+1)^2\,\allowbreak(x+\omega^2)\,\allowbreak
  (x^2+\omega x+\omega)$
& $(x+1)\,\allowbreak(x^4+x^3+x+\omega^2)$
& $d=d_{\mathrm{OpLC}}-3$ \\
\cline{1-1}\cline{3-10}

\nextcodeno
&
& $[42,7,23]^H_4$
& $1$
& $x^3+\omega x^2+\omega^2x+\omega^2$
& $x^5+\omega x^4+\omega x^3+\omega x^2+\omega x+\omega$
& $(x+\omega)\,\allowbreak
  (x^4+\omega^2x^3+x^2+\omega^2x+1)$
& $(x+\omega)\,\allowbreak(x+\omega^2)\,\allowbreak
  (x^2+x+\omega)$
& $(x+1)^2\,\allowbreak
  (x^4+\omega x^3+\omega^2x^2+\omega^2x+1)$
& $d=d_{\mathrm{BKLC}}-3$ \\
\noalign{\hrule height 1.2pt}
\noalign{\hrule height 0.5pt}

\nextcodeno
& 9
& $[54,4,36]^H_4$
& $(x+\omega)\,\allowbreak
  (x+\omega^2)\,\allowbreak
  (x^3+\omega)$
& $(x+\omega)\,\allowbreak
  (x+\omega^2)\,\allowbreak
  (x^3+\omega)\,\allowbreak
  (x^3+\omega x^2+\omega^2x+\omega^2)$
& $(x+\omega)^2\,\allowbreak
  (x+\omega^2)^3\,\allowbreak
  (x^3+\omega)$
& $(x+\omega)\,\allowbreak
  (x+\omega^2)\,\allowbreak
  (x^2+\omega x+1)\,\allowbreak
  (x^3+\omega)$
& $(x+\omega)\,\allowbreak
  (x+\omega^2)\,\allowbreak
  (x^3+\omega)\,\allowbreak
  (x^3+\omega x^2+\omega^2x+\omega)$
& $(x+1)\,\allowbreak
  (x+\omega)\,\allowbreak
  (x+\omega^2)\,\allowbreak
  (x^3+\omega)$
& $d=d_{\mathrm{OpLC}}-3$ \\
\cline{1-1}\cline{3-10}

\nextcodeno
&
& $[54,5,34]^H_4$
& $(x+\omega)\,\allowbreak(x^3+\omega)$
& $(x+\omega)\,\allowbreak
  (x^3+\omega)\,\allowbreak
  (x^3+\omega x^2+\omega^2x+\omega^2)$
& $(x+\omega)^3\,\allowbreak
  (x+\omega^2)^2\,\allowbreak
  (x^3+\omega)$
& $(x+\omega)\,\allowbreak
  (x^3+\omega)\,\allowbreak
  (x^3+x^2+\omega^2x+\omega)$
& $(x+\omega)\,\allowbreak
  (x^3+\omega)\,\allowbreak
  (x^3+\omega x^2+\omega^2x+\omega)$
& $(x+1)\,\allowbreak
  (x+\omega)\,\allowbreak
  (x^3+\omega)$
& $d=d_{\mathrm{OpLC}}-4$ \\
\cline{1-1}\cline{3-10}

\nextcodeno
&
& $[54,6,34]^H_4$
& $(x+1)\,\allowbreak
  (x+\omega)\,\allowbreak
  (x+\omega^2)$
& $(x+1)\,\allowbreak
  (x+\omega)\,\allowbreak
  (x+\omega^2)\,\allowbreak
  (x^3+\omega x^2+\omega^2x+\omega^2)$
& $(x+1)\,\allowbreak
  (x+\omega)\,\allowbreak
  (x+\omega^2)\,\allowbreak
  (x^4+\omega x^3+\omega^2x^2+\omega^2)$
& $(x+1)\,\allowbreak
  (x+\omega)^2\,\allowbreak
  (x+\omega^2)\,\allowbreak
  (x^4+\omega^2x^3+x^2+\omega^2x+1)$
& $(x+1)\,\allowbreak
  (x+\omega)^3\,\allowbreak
  (x+\omega^2)^2$
& $(x+1)\,\allowbreak
  (x+\omega)\,\allowbreak
  (x+\omega^2)^2\,\allowbreak
  (x^3+x^2+\omega^2x+\omega)$
& $d=d_{\mathrm{BKLC}}-2$ \\
\cline{1-1}\cline{3-10}

\nextcodeno
&
& $[54,7,32]^H_4$
& $(x+1)\,\allowbreak(x+\omega)$
& $(x+1)\,\allowbreak
  (x+\omega)\,\allowbreak
  (x^3+\omega x^2+\omega^2x+\omega^2)$
& $(x+1)\,\allowbreak
  (x+\omega)^3\,\allowbreak
  (x+\omega^2)\,\allowbreak
  (x^3+x^2+x+\omega^2)$
& $(x+1)\,\allowbreak
  (x+\omega)^2\,\allowbreak
  (x^4+\omega^2x^3+x^2+\omega^2x+1)$
& $(x+1)^2\,\allowbreak
  (x+\omega)\,\allowbreak
  (x^5+x^4+\omega^2x^3+\omega^2x^2+\omega x+1)$
& $(x+1)\,\allowbreak
  (x+\omega)^3\,\allowbreak
  (x^4+\omega x^3+\omega x^2+\omega^2x+1)$
& $d=d_{\mathrm{BKLC}}-3$ \\
\cline{1-1}\cline{3-10}

\nextcodeno
&
& $[54,8,31]^H_4$
& $x+\omega$
& $(x+\omega)\,\allowbreak
  (x^3+\omega x^2+\omega^2x+\omega^2)$
& $(x+\omega)^3\,\allowbreak
  (x^4+\omega x^3+x+1)$
& $(x+\omega)^2\,\allowbreak
  (x+\omega^2)\,\allowbreak
  (x^4+\omega^2x^3+x^2+\omega^2x+1)$
& $(x+\omega)\,\allowbreak
  (x^3+\omega^2x^2+1)\,\allowbreak
  (x^4+\omega x^3+x^2+1)$
& $(x+1)^2\,\allowbreak
  (x+\omega)^3\,\allowbreak
  (x^3+\omega^2x^2+x+\omega)$
& $d=d_{\mathrm{BKLC}}-2$ \\
\cline{1-1}\cline{3-10}

\nextcodeno
&
& $[54,9,27]^H_4$
& $1$
& $x^3+\omega x^2+\omega^2x+\omega^2$
& $(x+\omega)^2\,\allowbreak
  (x+\omega^2)\,\allowbreak
  (x^4+\omega x^3+x+1)$
& $(x+\omega)\,\allowbreak
  (x^4+\omega^2x^3+x^2+\omega^2x+1)$
& $(x^3+\omega^2x^2+1)\,\allowbreak
  (x^4+\omega x^3+x^2+1)$
& $(x+1)^3\,\allowbreak
  (x+\omega)\,\allowbreak
  (x^3+x+1)$
& $d=d_{\mathrm{BKLC}}-5$ \\

\noalign{\hrule height 1.2pt}
\noalign{\hrule height 0.5pt}
\end{tabular}

\caption{Hermitian LCD one-generator quasi-cyclic codes of index $6$.}
\label{tab:appendix-index6}
\end{table}

\end{landscape}
\clearpage


\begin{landscape}

\begin{table}[H]
\centering
\fontsize{8.0}{8.3}\selectfont
\renewcommand{\arraystretch}{1.00}
\setlength{\tabcolsep}{0.35pt}

\begin{tabular}{
|>{\centering\arraybackslash}m{0.52cm}|
>{\centering\arraybackslash}m{0.35cm}|
>{\centering\arraybackslash}m{1.95cm}|
>{\raggedright\arraybackslash}m{2.45cm}|
>{\raggedright\arraybackslash}m{2.45cm}|
>{\raggedright\arraybackslash}m{2.45cm}|
>{\raggedright\arraybackslash}m{2.45cm}|
>{\raggedright\arraybackslash}m{2.45cm}|
>{\raggedright\arraybackslash}m{2.45cm}|
>{\raggedright\arraybackslash}m{2.45cm}|
>{\centering\arraybackslash}m{1.35cm}|
}
\hline
\noalign{\hrule height 1.2pt}
No. & $m$ & $C$ & $g_{1}(x)$ & $g_{2}(x)$ & $g_{3}(x)$
& $g_{4}(x)$ & $g_{5}(x)$ & $g_{6}(x)$ & $g_{7}(x)$
& Remark \\
\noalign{\hrule height 1.2pt}

\nextcodeno
& 5
& $[35,4,24]^H_4$
& $x+1$
& $(x+1)\,\allowbreak(x+\omega)^2$
& $(x+1)\,\allowbreak(x^3+x^2+1)$
& $(x+1)\,\allowbreak
  (x^3+\omega^2x^2+\omega x+\omega^2)$
& $(x+1)\,\allowbreak(x+\omega^2)$
& $(x+1)\,\allowbreak
  (x^3+x^2+\omega^2x+\omega)$
& $(x+1)\,\allowbreak(x+\omega)$
& $d=d_{\mathrm{OpLC}}$ \\
\cline{1-1}\cline{3-11}

\nextcodeno
&
& $[35,5,21]^H_4$
& $1$
& $(x+\omega)^2$
& $x^4+\omega x^3+\omega x^2+x+\omega$
& $(x+\omega)\,\allowbreak
  (x+\omega^2)\,\allowbreak
  (x^2+x+\omega)$
& $x^4+x^3+x^2+\omega$
& $(x+1)\,\allowbreak
  (x+\omega)\,\allowbreak
  (x^2+x+\omega)$
& $(x+1)^2\,\allowbreak(x+\omega)$
& $d=d_{\mathrm{OpLC}}-3$ \\
\noalign{\hrule height 1.2pt}
\noalign{\hrule height 0.5pt}

\nextcodeno
& 7
& $[49,6,30]^H_4$
& $(x+1)\,\allowbreak(x^3+x+1)$
& $(x+1)\,\allowbreak
  (x^5+\omega x^4+x^3+\omega x^2+x+\omega^2)$
& $(x+1)\,\allowbreak
  (x+\omega)\,\allowbreak
  (x^2+x+\omega^2)$
& $(x+1)\,\allowbreak
  (x+\omega)\,\allowbreak
  (x+\omega^2)\,\allowbreak
  (x^2+x+\omega)$
& $(x+1)^2\,\allowbreak
  (x+\omega^2)\,\allowbreak
  (x^2+\omega x+\omega)$
& $(x+1)\,\allowbreak
  (x^2+\omega x+\omega)\,\allowbreak
  (x^3+x^2+x+\omega^2)$
& $x+1$
& $d=d_{\mathrm{BKLC}}-2$ \\
\cline{1-1}\cline{3-11}

\nextcodeno
&
& $[49,7,28]^H_4$
& $1$
& $(x+\omega)^3$
& $(x^3+x^2+x+\omega)\,\allowbreak
  (x^3+x^2+\omega x+\omega^2)$
& $(x+\omega^2)\,\allowbreak
  (x^5+\omega^2x^3+\omega x+1)$
& $x^4+\omega^2x^3+x^2+\omega^2x+\omega$
& $x^4+\omega x^3+x^2+1$
& $(x^2+\omega^2x+1)\,\allowbreak
  (x^2+\omega^2x+\omega^2)$
& $d=d_{\mathrm{OpLC}}-4$ \\
\noalign{\hrule height 1.2pt}
\noalign{\hrule height 0.5pt}

\nextcodeno
& 9
& $[63,4,42]^H_4$
& $(x+1)\,\allowbreak
  (x+\omega^2)\,\allowbreak
  (x^3+\omega)$
& $(x+1)\,\allowbreak
  (x+\omega)^3\,\allowbreak
  (x+\omega^2)\,\allowbreak
  (x^3+\omega)$
& $(x+1)\,\allowbreak
  (x+\omega^2)\,\allowbreak
  (x^2+\omega^2x+1)\,\allowbreak
  (x^3+\omega)$
& $(x+1)\,\allowbreak
  (x+\omega)\,\allowbreak
  (x+\omega^2)\,\allowbreak
  (x^3+\omega)$
& $(x+1)\,\allowbreak
  (x+\omega^2)^2\,\allowbreak
  (x^3+\omega)$
& $(x+1)\,\allowbreak
  (x+\omega^2)\,\allowbreak
  (x^3+\omega)\,\allowbreak
  (x^3+x^2+1)$
& $(x+1)\,\allowbreak
  (x+\omega^2)\,\allowbreak
  (x^2+\omega x+\omega)\,\allowbreak
  (x^3+\omega)$
& $d=d_{\mathrm{OpLC}}-5$ \\
\cline{1-1}\cline{3-11}

\nextcodeno
&
& $[63,5,42]^H_4$
& $(x+1)\,\allowbreak(x^3+\omega^2)$
& $(x+1)\,\allowbreak
  (x+\omega)\,\allowbreak
  (x^3+\omega^2)$
& $(x+1)\,\allowbreak
  (x^3+\omega^2)\,\allowbreak
  (x^4+\omega^2x^3+\omega x^2+\omega)$
& $(x+1)\,\allowbreak
  (x+\omega)^2\,\allowbreak
  (x^3+\omega^2)$
& $(x+1)\,\allowbreak
  (x^3+\omega^2)\,\allowbreak
  (x^4+\omega^2x^3+x^2+\omega^2x+\omega)$
& $(x+1)\,\allowbreak
  (x^2+x+\omega)\,\allowbreak
  (x^2+\omega x+\omega)\,\allowbreak
  (x^3+\omega^2)$
& $(x+1)\,\allowbreak
  (x^2+x+\omega^2)\,\allowbreak
  (x^3+\omega^2)$
& $d=d_{\mathrm{BKLC}}-2$ \\
\cline{1-1}\cline{3-11}

\nextcodeno
&
& $[63,6,40]^H_4$
& $(x+1)\,\allowbreak
  (x+\omega)\,\allowbreak
  (x+\omega^2)$
& $(x+1)\,\allowbreak
  (x+\omega)^4\,\allowbreak
  (x+\omega^2)$
& $(x+1)\,\allowbreak
  (x+\omega)\,\allowbreak
  (x+\omega^2)\,\allowbreak
  (x^4+x^3+x^2+\omega)$
& $(x+1)^3\,\allowbreak
  (x+\omega)\,\allowbreak
  (x+\omega^2)\,\allowbreak
  (x^2+\omega^2x+\omega^2)$
& $(x+1)\,\allowbreak
  (x+\omega)\,\allowbreak
  (x+\omega^2)^2\,\allowbreak
  (x^4+\omega^2x^3+x+\omega)$
& $(x+1)\,\allowbreak
  (x+\omega)\,\allowbreak
  (x+\omega^2)^3$
& $(x+1)\,\allowbreak
  (x+\omega)\,\allowbreak
  (x+\omega^2)^3\,\allowbreak
  (x^2+x+\omega^2)$
& $d=d_{\mathrm{OpLC}}-4$ \\
\cline{1-1}\cline{3-11}

\nextcodeno
&
& $[63,7,38]^H_4$
& $(x+1)\,\allowbreak(x+\omega^2)$
& $(x+1)\,\allowbreak
  (x+\omega)^3\,\allowbreak
  (x+\omega^2)$
& $(x+1)\,\allowbreak
  (x+\omega^2)\,\allowbreak
  (x^3+x^2+x+\omega)^2$
& $(x+1)\,\allowbreak
  (x+\omega)\,\allowbreak
  (x+\omega^2)\,\allowbreak
  (x^5+\omega x^4+x^2+\omega x+1)$
& $(x+1)^3\,\allowbreak
  (x+\omega^2)\,\allowbreak
  (x^4+\omega^2x^3+x^2+\omega^2x+\omega)$
& $(x+1)^3\,\allowbreak
  (x+\omega^2)\,\allowbreak
  (x^4+\omega x^3+x^2+1)$
& $(x+1)\,\allowbreak
  (x+\omega^2)^3\,\allowbreak
  (x^2+x+\omega^2)$
& $d=d_{\mathrm{OpLC}}-5$ \\
\cline{1-1}\cline{3-11}

\nextcodeno
&
& $[63,8,36]^H_4$
& $x+1$
& $(x+1)\,\allowbreak(x+\omega)^6$
& $(x+1)^2\,\allowbreak
  (x^2+\omega x+\omega)\,\allowbreak
  (x^3+x^2+1)$
& $(x+1)\,\allowbreak
  (x+\omega)^5\,\allowbreak
  (x^2+x+\omega)$
& $(x+1)\,\allowbreak
  (x^4+\omega^2x^3+x^2+\omega^2x+\omega)$
& $(x+1)^2\,\allowbreak
  (x^4+\omega x^3+x^2+1)$
& $(x+1)\,\allowbreak(x^2+x+\omega^2)$
& $d=d_{\mathrm{OpLC}}-6$ \\
\cline{1-1}\cline{3-11}

\nextcodeno
&
& $[63,9,35]^H_4$
& $1$
& $(x+\omega)^6$
& $(x+1)\,\allowbreak
  (x^2+\omega x+\omega)\,\allowbreak
  (x^3+x^2+1)$
& $(x+\omega)^5\,\allowbreak
  (x^2+x+\omega)$
& $(x+\omega^2)\,\allowbreak
  (x^2+\omega x+1)\,\allowbreak
  (x^2+\omega x+\omega)\,\allowbreak
  (x^3+\omega x+1)$
& $(x+1)\,\allowbreak
  (x^6+\omega x^5+\omega^2x^4+x^2+x+\omega)$
& $(x+\omega^2)^2\,\allowbreak
  (x^2+x+\omega^2)$
& $d=d_{\mathrm{BKLC}}-3$ \\

\noalign{\hrule height 1.2pt}
\noalign{\hrule height 0.5pt}
\end{tabular}

\caption{Hermitian LCD one-generator quasi-cyclic codes of index $7$.}
\label{tab:appendix-index7}
\end{table}

\end{landscape}
\clearpage

\end{document}